\RequirePackage{etex}
\documentclass[a4paper,twocolumn,10pt,unpublished,floatfix]{quantumarticle}
\pdfoutput=1

\usepackage{amsfonts}
\usepackage{amsthm}
\usepackage{amsmath}
\usepackage[english]{babel}
\usepackage{braket}
\usepackage{bm}
\usepackage{geometry}
\usepackage{amssymb}
\usepackage{lettrine}
\usepackage{tikz}
\usepackage{quantikz}
\usetikzlibrary{
	backgrounds,
	arrows,
	quotes,
	angles,
	positioning,
	decorations.pathreplacing,
	shapes.symbols}
\usepackage{tikz-network}
\usepackage{graphicx}
\usepackage[singlelinecheck=false]{subcaption} 

\usepackage[classfont=sanserif,langfont=sanserif,funcfont=sanserif]{complexity}
\usepackage{enumitem}
\usepackage{multirow}
\usepackage[createShortEnv]{proof-at-the-end}
\usepackage{csquotes}
\usepackage{epigraph}
\usepackage{anyfontsize}
\usepackage{pgfplots}
\pgfplotsset{compat=newest}
\usepgfplotslibrary{groupplots}
\usepgfplotslibrary{dateplot}
\usepackage{refcount}
\usepackage[ruled,lined]{algorithm2e}
\SetKwComment{Comment}{$\triangleright$ }{}
\usepackage{float}
\usepackage{mathrsfs}

\usepackage{silence}
\usepackage{thm-restate}

\usepackage[T1]{fontenc}
\vbadness=\maxdimen

\PassOptionsToPackage{numbers}{natbib}
\usepackage{natbib}
\usepackage{latexsym}

\usepackage{varwidth}
\tikzset{
    max width/.style args={#1}{
        execute at begin node={\begin{varwidth}{#1}},
        execute at end node={\end{varwidth}}
    }
}

\usepackage{bookmark}

\newclass{\PostBPP}{PostBPP}
\newclass{\SampP}{SampP}
\newclass{\SampBQP}{SampBQP}
\newclass{\FBPP}{FBPP}
\newclass{\BosonFP}{BosonFP}
\newclass{\BosonP}{BosonP}
\newclass{\CompClassA}{A}
\newclass{\CompClassB}{B}

\newtheorem{theorem}{Theorem}

\newtheorem{corollary}{Corollary}
\newtheorem{proposition}{Proposition}
\newtheorem{observation}{Observation}
\newtheorem*{proposition*}{Proposition}
\newtheorem{lemma}{Lemma}
\newtheorem*{lemma*}{Lemma}

\let\epsilon\varepsilon

\usepackage{apptools}
\definecolor{Navy}{HTML}{2874A6}
\definecolor{DarkGreen}{HTML}{117A65}
\definecolor{Green}{RGB}{115, 198, 182}
\definecolor{ultramarine}{RGB}{84,189,220}
\definecolor{violet}{RGB}{
	145,106,184
	}
\definecolor{darkviolet}{RGB}{
	99, 56, 142 
	}
\definecolor{lightviolet}{RGB}{
	158,136,180
	}

\definecolor{darkred}{RGB}{
	192, 57, 43
	}

\definecolor{darkgreen}{RGB}{
	39, 174, 96
	}

\hypersetup{
	colorlinks = true,
	allcolors = {darkviolet}, 
}

\usepackage[nameinlink]{cleveref} 
\crefname{chapter}{Chapter}{Chapters}
\crefname{section}{Section}{Sections}
\crefname{equation}{Equation}{Equations}
\crefname{observation}{Observation}{Observations}
\crefname{figure}{Figure}{Figures}
\crefname{table}{Table}{Tables}
\crefname{appendix}{Appendix}{Appendices}
\crefname{theorem}{Theorem}{Theorems}
\crefname{corollary}{Corollary}{Corollaries}
\crefname{lemma}{Lemma}{Lemmas}
\crefname{proposition}{Proposition}{Propositions}
\crefname{definition}{Definition}{Definitions}
\crefname{algorithm}{Algorithm}{Algorithms}
\crefname{thm}{Theorem}{Theorems}

\crefformat{footnote}{#2\footnotemark[#1]#3}

\let\autoref\cref

\newcommand{\bs}[1]{\boldsymbol{#1}}
\newcommand{\B}{\bs{B}}
\renewcommand{\b}{\bs{b}}
\renewcommand{\a}{\bs{a}}
\newcommand{\s}{\bs{s}}
\newcommand{\x}{\bs{x}}
\newcommand{\y}{\bs{y}}

\newcommand{\UU}{\mathcal{U}}

\newcommand{\CC}{\mathcal{C}}
\newcommand{\OO}{\mathcal{O}}

\newcommand{\II}{\mathcal{I}}
\newcommand{\GG}{\mathcal{G}}

\renewcommand{\EEE}{\mathscr{E}}

\newcommand{\Ng}{N^\GG}
\newcommand{\Ngstar}{N^{\GG*}}
\newcommand{\Ngp}{N^\GG_{post}}
\newcommand{\Nk}{N^{{\textrm{\tiny KLM}}}}
\newcommand{\Nkstar}{N^{\textrm{\tiny{KLM}*}}}
\newcommand{\Nkp}{N^{\textrm{\tiny KLM}}_{post}}

\newcommand{\est}[1]{\mathop{\EEE\left(#1\right)}}
\newcommand{\diag}{\mathop{\mathrm{diag}}}

\newcommand{\pr}[2][]{
	\mathop{
		\ifx &#1&
		\mathrm{Pr}
		\else
			\underset{#1}{\mathrm{Pr}}
		\fi
		\left[#2\right]}
}

\newcommand{\e}[2][]{
	\mathop{
		\ifx &#1&
			\mathbb{E}
		\else
			\underset{#1}{\mathbb{E}}
		\fi
		\left[#2\right]}
}

\newcommand{\ie}{\emph{i.e.}, }
\newcommand{\eg}{\emph{e.g.}, }

\newcommand{\per}[1]{\mathop{\mathsf{Per}\left(#1\right)}}

\newcommand{\abs}[1]{\left|#1\right|}

\newcommand{\bigo}[1]{\OO\left(#1\right)}
\newcommand{\toc}{\leadsto}

\title{Connecting quantum circuit amplitudes and matrix permanents through polynomials}
\author[1]{Hugo Thomas}
\email{hugo.thomas@quandela.com}
\author[1]{Pierre-Emmanuel Emeriau}
\author[1]{Rawad Mezher}
\affil[1]{Quandela SAS, 7 Rue Léonard de Vinci, 91300 Massy, France}

\begin{document}

\maketitle

\begin{abstract}
In this paper, we strengthen the connection between qubit-based quantum circuits and photonic quantum computation. Within the framework of circuit-based quantum computation, the sum-over-paths interpretation of quantum probability amplitudes leads to the emergence of sums of exponentiated polynomials. In contrast, the matrix permanent is a combinatorial object that plays a crucial role in photonic by describing the probability amplitudes of linear optical computations.
To connect the two, we introduce a general method to encode an $\mathbb F_2$-valued polynomial with complex coefficients into a graph, such that the permanent of the resulting graph's adjacency matrix corresponds directly to the amplitude associated the polynomial in the sum-over-path framework. 
This connection allows one to express quantum amplitudes arising from qubit-based circuits as permanents, which can naturally be estimated on a photonic quantum device. 

\end{abstract}

\section{Introduction}

Quantum computation is often pictured in the gate-based model where
computations are described by quantum logic gates. Quantum
computing (QC) based on qubits is predominant as its development has
mirrored that of classical computer science, which relies on bits as the
information carrier. 

Photonic quantum computing is a different yet promising avenue for
realising a quantum computer
\cite{maring_versatile_2024,bartolucci_fusion-based_2023}, since photonic
approaches are modular and flexible, and that photons hardly suffer from
decoherence. Photonic QC is not naturally described by qubit logic, though.
Rather, Aaronson and Arkhipov \cite{aaronson_computational_2011} showed that
output probability amplitudes of photons in a linear optical interferometer
(BosonSampling) can be expressed as permanents of sub-matrices of the linear
interferometer. More precisely, the quantum probability amplitudes resulting from a linear
optical computation on $m$ modes are -- for input and output occupancy $\bs s = (s_1, \cdots,
s_m)$ and $\bs t = (t_1, \cdots, t_m)$, where $s_i$ (resp. $t_i$) is the number
of photons in input (resp. output) mode $i$, and a unitary matrix $\mathcal U
\in SU(m)$ -- of the form 
\begin{equation}
  \braket{\bs t | \mathcal U |\bs u} \propto \per{\mathcal U_{\bs s, \bs t}}.
\end{equation}
Computing the matrix permanent is in all generality a so-called \#\P-hard problem -- that
is, a problem that amounts to counting the number of solutions to an \NP-hard problem; in
\cite{valiant_complexity_1979} it is shown that it is a $\#\P$-complete problem
for binary matrices. 

Sums of exponentiated polynomials\footnote{Also called \textit{exponential sums} in \cite{bu_classical_2022,cai_tractable_2010}.} appear in many domains ranging from number
theory \cite{bombieri_exponential_1966,korobov_weyls_1992} and computational
complexity theory \cite{stockmeyer_complexity_1983} to Feynman's interpretation
of quantum mechanics \cite{feynman_quantum_1966}. In \cite{dawson_quantum_2004},
the authors relied on a discretised version of Feynman's path-integral
formulation to show that computing quantum circuits probability amplitudes is
equivalent to compute such sums. More precisely, given a quantum circuit $\CC$
acting on $n$ qubits, the probability amplitude of a given input-output pair
$\bs a, \bs b \in \{0, 1\}^n$ is given by the sum
\begin{equation}
  \braket{\bs b | \CC | \bs a} \propto \sum_{\bs x \in \{0, 1\}^v} e^{i f(\bs x)},
\end{equation}
where $f$ is an $\mathbb F_2$-valued polynomial with complex coefficients and $v$
is the number of qubits in the circuit plus the number of gates creating
coherence in the circuit (\eg Hadamard gates). From this result follows that
computing quantum circuits probability amplitudes is also a \#\P-hard
function \cite{ehrenfeucht_computational_1990}.

Intrigued by the suggestion at the end of \cite{rudolph_simple_2009} regarding
the potential estimation of quantum circuits amplitude via classical techniques
coming from the linear optical picture
\cite{gurvits_complexity_2005,aaronson_generalizing_2012}, we aim to derive a
formal mapping between gate-based computation and linear optical computation for
amplitude estimation. The two seemingly different models of computation yields
quantum probability amplitudes the computation of which lie in the same
complexity class. This raises the question of the precise relationship between
these two models.

\paragraph{Our contributions.} From a $n$-variable polynomial $f$ with binary
variables and complex coefficients, we show how to construct a graph $G$ with
adjacency matrix $\mathcal G$ such that
\begin{equation}\label{eq:expSumPermanent}
  \sum_{x \in \{0, 1\}^n} e^{i f(\x)} = \per{\mathcal G},
\end{equation}
in time polynomial in the number of clauses of $f$. That is to say, the permanent
computes the sum of the exponential of a polynomial $f$ evaluated over all its
domain.  We go beyond already existing techniques
\cite{valiant_complexity_1979,rudolph_simple_2009,blaser_complexity_2012} by
providing a systematic recipe for constructing such a graph. To do so, each
clause of the polynomial is encoded onto a small graph, called a \emph{clause-gadget},
according to its degree. For a fixed degree, the graph gadget is parameterized
by the clause coefficient. While we find that designing a graph gadget is a
computationally hard problem\footnote{Although note that it is a compute once and store forever problem.}, we perform the computation for clauses of degree
up to 3 in \autoref{sec:example}. This encoding finds uses in quantum circuits
probability amplitude computation. It was shown in \cite{mezher_solving_2023} that it is possible to encode the
adjacency matrix of a graph $G(V,E)$ onto a unitary matrix $\mathcal U$ such
that 
\begin{equation*}
    \braket{\bs 1_n | \mathcal U |\bs 1_n} \propto \per{\mathcal G},
\end{equation*}
where $\ket{\bs 1_n} =\ket{1^{\otimes n}0^{\otimes{m-n}}}$. Thus, given the
graph $G_\CC(V,E)$ encoding the amplitude $\braket{\b | \CC | \a}$ of a quantum
circuit $\CC$,  for $\a, \b \in \{0, 1\}^n$, our encoding allows
one to construct a linear optical circuit $\mathcal V$ such that 
\begin{equation*}
    \braket{\bs 1_n | \mathcal V |\bs 1_n} \propto \braket{\bs b | \CC | \bs a }.
\end{equation*}
While Rudolph's encoding
\cite{rudolph_simple_2009} focuses on a particular gate set and does not provide
a generic recipe for designing suitable gadgets, we formalize the encoding and
provide a systematic recipe for constructing such gadgets.
We do so by deriving a procedure that checks the existence of a gadget and
guarantees its construction in \autoref{pr:gadgetExistence}. We discuss
potential applications of our encoding in \autoref{sec:applications}.

\begin{figure*}[t]
    \centering
    \begin{subfigure}[t]{.33\textwidth}
      \centering
        \begin{tikzpicture}[->,>=stealth',shorten >=1pt,thick, scale=.7]
            \node[draw, circle] (1) at (0,0) {1};
            \node[draw, circle] (2) at (1,1.4) {2};
            \node[draw, circle] (3) at (2, 0) {3};
            \path (1) edge [out=265,in=215,looseness=8] node[below left] {a} (1);
            \path (2) edge [out=115,in=65,looseness=8] node[above] {b} (2);
            \draw[->, bend left = 10] (1) edge["d"] (2);
        \path (3) edge [out=340,in=290,looseness=8] node[below right] {c} (3);
            \draw[->, bend left = 10] (2) edge["e"] (3);
            \draw[->, bend left = 10] (3) edge["f"] (1);
        \end{tikzpicture}
      \caption{A 3-vertex graph denoted $G$.}
      \label{fig:coverCompleteGraph}
    \end{subfigure}\hfill%
    \begin{subfigure}[t]{.33\textwidth}
      \centering
        \begin{tikzpicture}[->,>=stealth',shorten >=1pt,thick, scale=.7]
            \node[draw, circle] (1) at (0,0) {1};
            \node[draw, circle] (2) at (1,1.4) {2};
            \node[draw, circle] (3) at (2, 0) {3};
        
            \path (1) edge [out=265,in=215,looseness=8] node[below left] {a} (1);
            \path (2) edge [out=115,in=65,looseness=8] node[above] {b} (2);
            \path (3) edge [out=340,in=290,looseness=8] node[below right] {c} (3);
        \end{tikzpicture}
      \caption{A cycle cover of $G$ of weight $abc$.}
      \label{fig:coverCompleteGraphCover1}
    \end{subfigure}\hfill%
    \begin{subfigure}[t]{.33\textwidth}
      \centering
        \begin{tikzpicture}[->,>=stealth',shorten >=1pt,thick, scale=.7]
            \node[draw, circle] (1) at (0,0) {1};
            \node[draw, circle] (2) at (1,1.4) {2};
            \node[draw, circle] (3) at (2, 0) {3};

            \node[] (4) at (0,-1.5) {};
        
            \draw[->, bend left = 10] (1) edge["d"] (2);
            \draw[->, bend left = 10] (2) edge["e"] (3);
            \draw[->, bend left = 10] (3) edge["f"] (1);
        \end{tikzpicture}
      \caption{Another cycle cover of $G$ of weigth $def$.}
      \label{fig:coverCompleteGraphCover2}
    \end{subfigure}%
    \caption{The leftmost graph of \autoref{fig:coverCompleteGraph} is a graph
     with 3 vertices. The two graphs depicted in
     \autoref{fig:coverCompleteGraphCover1,fig:coverCompleteGraphCover2} are the
     two cycles covers graph of \autoref{fig:coverCompleteGraph} -- so that if
     $G$ denotes the leftmost graph, $\per{G} = abc + def$.}
    \label{fig:cycleCoverIllustration}
  \end{figure*}
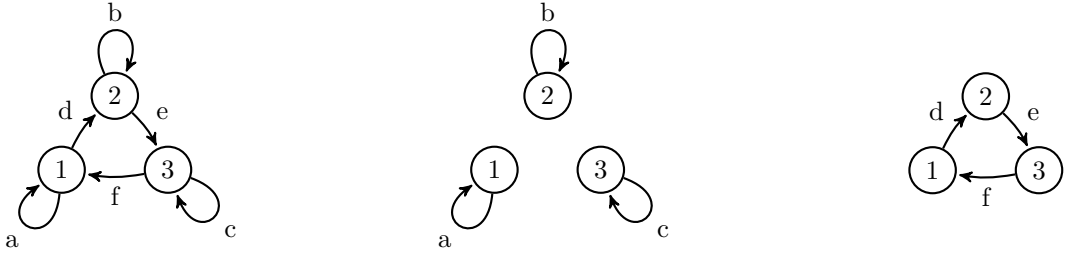

\section{Overview of the framework}\label{sec:framework}

Foremost, we start with a brief overview of the framework. We recall in
\autoref{app:linalg} the linear algebraic tools used throughout this work.

In \cite{rudolph_simple_2009}, the author shows how to encode a family of
specific polynomials that encode an amplitude of a quantum circuit onto a
graph, such that the permanent of the adjacency matrix of the graph equals the
amplitude. The technique relies on a seminal work from Valiant
\cite{valiant_complexity_1979}, where such a graph encoding is used to show that
computing the permanent of a binary matrix is a $\#\P$-complete problem.
Formally, in linear algebra, the permanent of a square matrix $A = (a_{i,j})_{0
\leq i,j\leq n}$ of order $n$, which we write $\per A$, is defined as the
following function of its entries \cite{marcus_permanents_1965}:
\begin{equation}\label{eq:defPermanent}
\per A = \sum_{\sigma \in S_n} \prod_{i = 1}^n a_{i, \sigma(i)},
\end{equation}
where the sum extends over all elements of the symmetric group $S_n$. Despite
its apparent similarities with the determinant, the essence of the permanent is
counting and it has, unlike the determinant, no known geometrical
interpretation. From a graph-theoretic point of view, the permanent of the
adjacency matrix of a directed graph -- \ie the direction of the edges matters
-- is equal to the sum of the weights of all \emph{cycle covers} of the graph
(defined shortly after). Aaronson proved that computing the permanent is
$\#\P$-hard in the general case using linear optics techniques
\cite{aaronson_linear-optical_2011}.

We use the notation $G(V,E)$ to denote a graph with vertex set $V$, and edge set
$E$, or $G$ for simplicity if these sets are irrelevant, and we denote its
adjacency matrix by $\mathcal G$, \ie the same symbol with calligraphic font.

We consider hereafter directed and weighted graphs. In such a graph, a
\emph{cycle} is a closed path. A \emph{cycle cover} is a set of vertex-disjoint
cycles that covers \emph{all} vertices, \ie no vertex is left alone. In all
generality, a graph admits several cycle covers and given one of them, its
weight is the product of the weights of the edges involved in that particular
cover. We show examples of graph cycle covers in
\autoref{fig:cycleCoverIllustration}.
 
Providing a graph $G(V,E)$ satisfying \autoref{eq:expSumPermanent} is achieved
by encoding each monomial (also identified as clauses) of the polynomial $f$
onto a small graph, the so-called \emph{clause gadgets}, and those gadgets are
linked together with cycles representing the variables they share to form the
graph. From the perspective of the graph, the permanent of its adjacency matrix
equals the sum of the weights of all its cycle covers. The permanent of the
graph relates to evaluating the polynomial as follows. In a particular cycle
cover, the absence of a cycle corresponding to a variable corresponds to the
fact that the variable is assigned the value 1, and its presence corresponds to
its assignment to 0. We hereafter use the following coloring convention: the
blue cycles encode the variables while linear, quadratic and cubic clauses are
respectively represented by red, green and purple graph-gadgets.

For the sake of completeness, we first explain how exponential sums may appear
in quantum computation. In \cite{dawson_quantum_2004} a method -- known as the
\emph{sum-over-paths} formalism -- is given to encode the quantum amplitudes of
\{Hadamard, Toffoli\} circuits as the gap of a low-degree polynomial over
$\mathbb F_2$. Precisely, for a $q$-qubit circuit $\CC$ comprising $h$ Hadamard
gates, it is shown how to construct a polynomial ${f \in \mathbb F_2[x_1,\cdots,
x_n]}$ with $n = q + h$ variables such that
\begin{equation}\label{eq:amplitudeSumValidAssigment}
    \braket{\b | \CC | \a}
        = \sum_{\substack{\x \in \{0, 1\}^{n}}}\frac{(-1)^{f(\x)}}{\sqrt 2^h},
\end{equation}
where $\ket \a$ and $\ket \b$ are arbitrary input and output states, by applying
the following procedure. In the rest, we will refer to $\braket{\b | \CC | \a}$
as the \emph{amplitude of interest} with the multiset notation $\x = (x_1,
\cdots, x_n)$, and consider the $a_i$'s and $b_j$'s as input and output
variables respectively. The $a_i$'s travel along the their respective qubit
line. Whenever a variable $x$ goes through a Hadamard gate, a new one is created
(\autoref{fig:hGate}), while when a variable $z$ traverses the target of a
Toffoli gate, no variable is created but it is replaced by $z \oplus xy$ where
$x$, $y$ are the variables living on the two control lines that remain unchanged
(\autoref{fig:toffoliGate}).
\begin{figure}[hbtp!]
    \centering
    \begin{subfigure}{.5\linewidth}
      \centering
      \begin{quantikz}
        \\
        \lstick{$x$} & \gate{H} & \qw \rstick{$y$} \\
        \end{quantikz}
      \caption{Hadamard gate.}
      \label{fig:hGate}
    \end{subfigure}\hfill%
    \begin{subfigure}{.5\linewidth}
      \centering
      \begin{quantikz}
        \lstick{$x$} & \ctrl{1} & \qw \rstick{$x$}\\
        \lstick{$y$} & \ctrl{1} & \qw \rstick{$y$}\\
        \lstick{$z$} & \targ{} & \qw \rstick{$z \oplus xy$}
        \end{quantikz}
      \caption{Toffoli gate.}
      \label{fig:toffoliGate}
    \end{subfigure}
    \caption{Rules to follow to act on the variables.}
    \label{fig:actionOnVariables}
\end{figure}
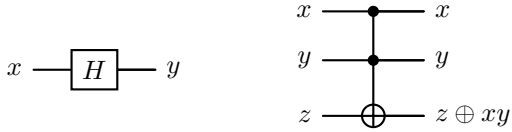
The output bit values $B_j(\x)$ are polynomials over $\mathbb F_2[x_1, \cdots,
x_n]$, where $n$ is the total number of path variables. Any assignment of $\x$
is valid if it satisfies the condition ${B_j(\x) = b_j}$ for all $j$, which we
write $\B(\x)= \b$. Once all the path variables are created, the circuit is
associated to the polynomial constructed by taking the sum over $\mathbb F_2$ of
the product of every pairs of variables on either side of a Hadamard gate,
namely
$$
f(x) = \sum_{\substack{\text{Hadamard}\\\text{gates $h$}}}
          (\text{input var. of $h$})(\text{output var. of $h$}) \,,
$$
as this generalizes to all Hadamard gates the identity 
\begin{equation}
    \braket{x | H | y} = \frac{1}{\sqrt 2}e^{i \pi xy}.
\end{equation}
Taking the sum over all possible assignments of the variables gives the
exponential sum of \autoref{eq:amplitudeSumValidAssigment}.
\begin{figure}[ht!]
    \centering
    \input{img/ht_circuit_example.tex}
    \caption{ Instance of $\{$Hadamard, Toffoli$\}$ circuit, with superimposed
      variables following the rules of \autoref{fig:actionOnVariables}. The
      vector of variables is ${\x = (x_0, \cdots, x_{12})}$. The input state is
      $\ket{\a} = \ket{x_0, x_1, x_2}$ and the vector of conditions is ${\B(\x)
      = (x_6, x_9, x_8)^T}$. The polynomial associated to the circuit is $f(\x)
      = x_0x_3 \oplus x_1x_4 \oplus x_2x_5 \oplus x_3x_6 \oplus x_4x_7 \oplus
      x_3x_5x_7 \oplus x_5x_8 \oplus x_7x_9$. The pair of Hadamard gates that
      create the variables $x_7$ and $x_9$ amounts to the identity, but enforce
      the polynomial to have degree (at most) 3. }
    \label{fig:HTCircuitExample}
\end{figure}
From this quantum amplitude encoding to polynomial, the contribution of
\cite{rudolph_simple_2009} is a procedure to construct a graph $G$ such that
\begin{equation}\label{eq:amplitudeFromPermanent}
\braket{\b | \CC | \a} = \frac{\per{\mathcal G}}{\sqrt 2^h},
\end{equation}
for circuit $\CC$ built upon the same gate set. Denote by $\CC$ the circuit of
\autoref{fig:HTCircuitExample}, the graph encoding of $\braket{111|\CC|000}$ is
illustrated in \autoref{fig:HTGraphExample}. 
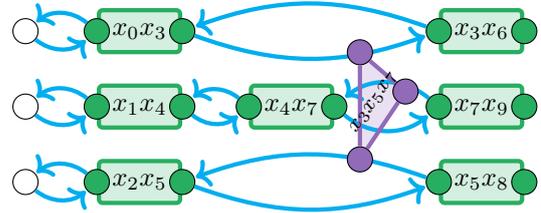
\begin{figure}[ht!]
  \centering
  \begin{tikzpicture}[
    scale=1
]
   \tikzset{dot/.style={circle,draw=, fill=#1}}
   \tikzset{arr/.style={color=cyan, ultra thick, -> = stealth,shorten > = .5pt,}}
   \tikzset{rec/.style={rectangle,draw, ultra thick, color=darkgreen, fill=darkgreen!20, text = black, rounded corners=0.05cm, inner sep=5pt}}

   \node[rec] (x2x5) at (0,0) {$x_2x_5$};
   \node[rec] (x1x4) at (0,1) {$x_1x_4$};
   \node[rec] (x0x3) at (0,2) {$x_0x_3$};
   \node[rec] (x5x8) at (4.5,0) {$x_5x_8$};
   \node[rec] (x7x9) at (4.5,1) {$x_7x_9$};
   \node[rec] (x3x6) at (4.5,2) {$x_3x_6$};
   \node[rec] (x4x7) at (2,1) {$x_4x_7$};

   \node[dot=darkgreen] (x2x5east) at (x2x5.east) {};
   \node[dot=darkgreen] (x2x5west) at (x2x5.west) {};
   \node[dot=darkgreen] (x1x4east) at (x1x4.east) {};
   \node[dot=darkgreen] (x1x4west) at (x1x4.west) {};
   \node[dot=darkgreen] (x0x3east) at (x0x3.east) {};
   \node[dot=darkgreen] (x0x3west) at (x0x3.west) {};
   \node[dot=darkgreen] (x5x8east) at (x5x8.east) {};
   \node[dot=darkgreen] (x5x8west) at (x5x8.west) {};
   \node[dot=darkgreen] (x7x9east) at (x7x9.east) {};
   \node[dot=darkgreen] (x7x9west) at (x7x9.west) {};
   \node[dot=darkgreen] (x3x6east) at (x3x6.east) {};
   \node[dot=darkgreen] (x3x6west) at (x3x6.west) {};
   \node[dot=darkgreen] (x4x7east) at (x4x7.east) {};
   \node[dot=darkgreen] (x4x7west) at (x4x7.west) {};

   \path[arr] (x2x5east) edge [bend right=20]  (x5x8west);
   \path[arr] (x5x8west) edge [bend right=20]  (x2x5east);

   \path[arr] (x1x4east) edge [bend right=40]  (x4x7west);
   \path[arr] (x4x7west) edge [bend right=40]  (x1x4east);

   \path[arr] (x4x7east) edge [bend right=40]  (x7x9west);
   \path[arr] (x7x9west) edge [bend right=40]  (x4x7east);

   \path[arr] (x0x3east) edge [bend right=20]  (x3x6west);
   \path[arr] (x3x6west) edge [bend right=20]  (x0x3east);

   \node[dot=violet] at (3.5, 1.2) (cx7) {};
   \node[dot=violet] at (2.9, 1.72) (cx3) {};
   \node[dot=violet] at (2.9, 0.3) (cx5) {};

   \begin{scope}[on background layer]  
    \fill [violet!20] (cx3.center) -- (cx7.center) -- (cx5.center) -- cycle;
   \end{scope}
   \draw[color=violet, ultra thick] (cx3) -- (cx5);
   \draw[color=violet, ultra thick] (cx7) -- (cx5);
   \draw[color=violet, ultra thick] (cx3) -- (cx7);
   \node[label={[rotate=55]\footnotesize$x_3x_5x_7$}] at (3.25, .82) () {};

   \node[dot=white] at (-1.5, 0) (zx2) {};
   \node[dot=white] at (-1.5, 1) (zx1) {};
   \node[dot=white] at (-1.5, 2) (zx0) {};

   \path[arr] (x0x3west) edge [bend right=40]  (zx0);
   \path[arr] (zx0) edge [bend right=40]  (x0x3west);

   \path[arr] (x1x4west) edge [bend right=40]  (zx1);
   \path[arr] (zx1) edge [bend right=40]  (x1x4west);
   \path[arr] (x2x5west) edge [bend right=40]  (zx2);
   \path[arr] (zx2) edge [bend right=40]  (x2x5west);

\end{tikzpicture}
  \caption{Graph corresponding to the circuit of \autoref{fig:HTCircuitExample},
  encoding input and output states $\ket{000}$ and $\ket{111}$ respectively.
  Recall that the circuit was associated to the polynomial $f(\x) = x_0x_3
  \oplus x_1x_4 \oplus x_2x_5 \oplus x_3x_6 \oplus x_4x_7 \oplus x_3x_5x_7
  \oplus x_5x_8 \oplus x_7x_9$. The green squares (resp. red triangles)
  represent the quadratic (resp. cubic) clause gadgets. The blue cycles
  correspond to the variables. We denote by $G$ the above graph and $\CC$
  the circuit of \autoref{fig:HTCircuitExample}, then it holds that
  $\braket{111|\CC|000}= \frac{\per{\mathcal G}}{8 \sqrt 2}$.}
  \label{fig:HTGraphExample}
\end{figure}
This encoding, however, imposes strong constraints on the structure of the
polynomial. Namely, the polynomial must be of degree at most 3, and the
variables that appear in a cubic clause must appear in exactly two quadratic
clauses. This constraint is satisfied by adding redundant pairs of Hamadard gate which, by construction, bound the degree of the polynomial. Hence, the proposal of \cite{rudolph_simple_2009} is tailored to encode
the polynomials giving the quantum amplitudes of circuits built upon the
Hadamard and Toffoli gates.

In the next section, we describe how to construct suitable clause gadgets for
polynomials without constraining their structure.

\section{General method for deriving clause gadgets}
\label{sec:method}

We encode polynomials of the form 
\begin{equation}\label{eq:polynomialForm}
    f(x_1, \cdots, x_n) = \sum_{S \subseteq [n]}
\theta_S \prod_{i \in S} x_i,
\end{equation}
with boolean variables $x_i$ and complex coefficients $\theta_S$ onto graphs. To
do so, each clause $\theta_S \prod_{i \in S} x_i$ is associated to a
graph-gadget, since using the identity
\begin{equation}\label{eq:exponentialRule}
    \sum_{\x} e^{i \sum_{S \subset [n]} \theta_S \prod_{i \in S} x_i} 
    = \sum_{\x} \prod_{S \subset [n]} e^{i \theta_S \prod_{i \in S} x_i},
\end{equation}
we find that we can treat each clause separately: observe that the left-hand side of
\autoref{eq:exponentialRule} has the same sum-of-products form as the permanent
of \autoref{eq:defPermanent}.

First in \cite{valiant_complexity_1979}, and later in \cite{rudolph_simple_2009}, was
shown that a gadget must satisfy the two following simple observations to be a
valid one:
\begin{observation}[\emph{nonzero} contributions]\label{ob:nonzero} If a cycle
cover is consistent with an assignment of the variables, then the weight of that
cycle cover must be the value of $f$ evaluated for that assignment of the
variables (see \autoref{fig:validCover}).
\end{observation}
\begin{observation}[\emph{zero} contributions]\label{ob:zero}
    The inconsistent cycle covers must have an overall weight of zero (see \autoref{fig:quadraticGadgetsCrossingPaths}).
\end{observation}
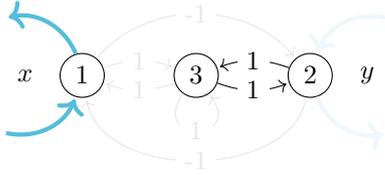
\begin{figure}[t!]
    \centering%
        \begin{tikzpicture}[
    -> = stealth, 
    shorten > = 1pt, 
    thin, 
    scale=1.5,
    every node/.append style={fill=white},
]


\node[draw, circle, inner sep = 0.7ex] at (0, 0) (1) {1};
\node[draw, circle, inner sep = 0.7ex] at (2, 0) (2) {2};
\node[draw, circle, inner sep = 0.7ex] at (1, 0) (3) {3};

\path[->, color=black!10] (1) edge [bend left=50] node [pos=0.5] {-1} (2);
\path[->, color=black!10] (2) edge [bend left=80] node {-1} (1);

\path[->, color=black!10] (1) edge [bend left=20] node {1} (3);
\path[->, color=black!10] (3) edge [bend left=20] node {1} (1);

\path[->, color=black] (2) edge [bend right=20] node {1} (3);
\path[->, color=black] (3) edge [bend right=20] node {1} (2);

\draw[color=black!10] (3) to [out=235, in=-55, looseness=6] node[] {1} (3);

\node[inner sep = 0.7ex] at (-0.5, 0) (20) {$x$};
\node[inner sep = 0.7ex] at (2.5, 0) (21) {$y$};

\node[] at (-0.75, 0.5) (4) {};
\node[] at (-0.75, -0.5) (5) {};
\path[color=ultramarine, ultra thick] (1) edge [bend right=40]  (4);
\path[color=ultramarine, ultra thick] (5) edge [bend right=40]  (1);

\node[] at (2.75, 0.5) (6) {};
\node[] at (2.75, -0.5) (7) {};
\path[color=ultramarine!10, ultra thick] (2) edge [bend right=40]  (7);
\path[color=ultramarine!10, ultra thick] (6) edge [bend right=40]  (2);
\end{tikzpicture}
        \caption{\emph{Valid} cycle cover (\autoref{ob:nonzero}) of the gadget
        representing the quadratic clauses of the polynomial from
        \cite{rudolph_simple_2009}. Here is illustrated the gadget $A_2(\pi)$
        (see \autoref{sec:example}). Say vertex 1 is connected to the variable
        $x$ and 2 to $y$, then the cycle cover assigns $x = 0$ and $y = 1$, and
        the contribution of the gadget is $(-1)^{xy} = 1$.}
        \label{fig:validCover}
\end{figure}
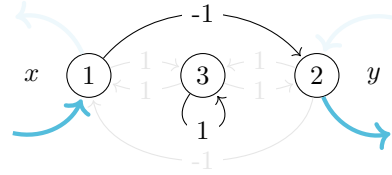
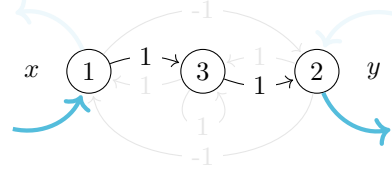
\begin{figure}[t!]
    \centering
    \begin{subfigure}{.45\textwidth}
        \centering
        \begin{tikzpicture}[
    -> = stealth, 
    shorten > = 1pt, 
    thin, 
    scale=1.5,
    every node/.append style={fill=white}
]


\node[draw, circle, inner sep = 0.7ex] at (0, 0) (1) {1};
\node[draw, circle, inner sep = 0.7ex] at (2, 0) (2) {2};
\node[draw, circle, inner sep = 0.7ex] at (1, 0) (3) {3};

\path[->] (1) edge [bend left=50] node [pos=0.5] {-1} (2);
\path[->, color=black!10] (2) edge [bend left=80] node {-1} (1);

\path[->, color=black!10] (1) edge [bend left=20] node {1} (3);
\path[->, color=black!10] (3) edge [bend left=20] node {1} (1);

\path[->, color=black!10] (2) edge [bend right=20] node {1} (3);
\path[->, color=black!10] (3) edge [bend right=20] node {1} (2);

\draw[] (3) to [out=235, in=-55, looseness=6] node[] {1} (3);

\node[inner sep = 0.7ex] at (-0.5, 0) (20) {$x$};
\node[inner sep = 0.7ex] at (2.5, 0) (21) {$y$};

\node[] at (-0.75, 0.5) (4) {};
\node[] at (-0.75, -0.5) (5) {};
\path[color=ultramarine!10, ultra thick] (1) edge [bend right=40]  (4);
\path[color=ultramarine, ultra thick] (5) edge [bend right=40]  (1);

\node[] at (2.75, 0.5) (6) {};
\node[] at (2.75, -0.5) (7) {};
\path[color=ultramarine, ultra thick] (2) edge [bend right=40]  (7);
\path[color=ultramarine!10, ultra thick] (6) edge [bend right=40]  (2);
\end{tikzpicture}
        \caption{Part of a crossing cycle.}
        \label{fig:crossPath1}
    \end{subfigure}\hfill%
    \begin{subfigure}{.45\textwidth}
        \centering
        \begin{tikzpicture}[
    -> = stealth, 
    shorten > = 1pt, 
    thin, 
    scale=1.5,
    every node/.append style={fill=white}
]


\node[draw, circle, inner sep = 0.7ex] at (0, 0) (1) {1};
\node[draw, circle, inner sep = 0.7ex] at (2, 0) (2) {2};
\node[draw, circle, inner sep = 0.7ex] at (1, 0) (3) {3};

\path[->, color=black!10] (1) edge [bend left=50] node [pos=0.5]{-1} (2);
\path[->, color=black!10] (2) edge [bend left=80] node {-1} (1);

\path[->] (1) edge [bend left=20] node {1} (3);
\path[->, color=black!10] (3) edge [bend left=20] node {1} (1);

\path[->, color=black!10] (2) edge [bend right=20] node {1} (3);
\path[->] (3) edge [bend right=20] node {1} (2);

\draw[color=black!10] (3) to [out=235, in=-55, looseness=6] node[] {1} (3);

\node[inner sep = 0.7ex] at (-0.5, 0) (20) {$x$};
\node[inner sep = 0.7ex] at (2.5, 0) (21) {$y$};

\node[] at (-0.75, 0.5) (4) {};
\node[] at (-0.75, -0.5) (5) {};
\path[color=ultramarine!10, ultra thick] (1) edge [bend right=40]  (4);
\path[color=ultramarine, ultra thick] (5) edge [bend right=40]  (1);

\node[] at (2.75, 0.5) (6) {};
\node[] at (2.75, -0.5) (7) {};
\path[color=ultramarine, ultra thick] (2) edge [bend right=40]  (7);
\path[color=ultramarine!10, ultra thick] (6) edge [bend right=40]  (2);
\end{tikzpicture}
        \caption{Part of another crossing cycle.}
        \label{fig:crossPath2}
    \end{subfigure}
    \caption{Invalid cycle cover: the cycles representing the variables are
    partially covered. The gadget design must be so that each crossing path can
    be paired with another with contribution of opposite sign to the cycle cover
    in a way which causes the necessary cancellation (\autoref{ob:zero}). In
    this instance, the contribution depicted in \autoref{fig:crossPath1} is
    cancelled by the one shown in \autoref{fig:crossPath2}.}
    \label{fig:quadraticGadgetsCrossingPaths}
\end{figure}

The clause gadget of a clause of degree $d$ is a graph with at least $d$
vertices, and the two above observations can be satisfied by adding a certain
number of additional \emph{inner} vertices if necessary.
\autoref{alg:generation} describes the procedure to explicitly write down all
the constraints a particular clause gadget must satisfy as a system of
polynomial equations.

Questions remain open: first and foremost, how do we find a solution to that
system? What is the minimal size of the graph associated to a given clause? To
answer these two questions, we extensively rely on Gröbner basis theory
\cite{cox_ideals_2013}. 

Designing clause gadgets involves solving a system of matrix permanents, whose
size depends on the degree of the clause to be encoded. State-of-the-art
algorithms for computing Gröbner bases, such as F5 \cite{faugere_new_2002}, have
exponential time complexity \cite{bardet_complexity_2015}. We believe that our
method is the most efficient way to design a gadget, as we draw up the minimal
set of constraints that a gadget must satisfy to be valid. Hence, we conjecture
that designing a suitable clause gadget is a computationally hard problem, the
requirement of computing permanents cannot be removed, as at least the permanent
of the gadget must equal the encoded amplitude.

\subsection{Clause gadget design}\label{sec:gateGadgetDesign}

This section is intended to explain how the methods of
\cite{valiant_complexity_1979,rudolph_simple_2009,blaser_complexity_2012} can be
formalized and how the finding of a clause-gadget can be systematized. The two
\autoref{ob:nonzero,ob:zero} of \autoref{sec:method} are suitably described
mathematically employing computational algebraic geometry techniques. Given a
gadget, we denote by \emph{outer vertices} all the vertices connected to the
rest of the graph, \ie the vertices associated to the variables, and by
\emph{inner vertices} the vertices only connected within the gadget. The inner
vertices will account for additional degrees of freedom in the weights the cycle
covers of the gadget so that all the constraints (see the rest of the section)
are satisfied. For instance, in
\autoref{fig:quadraticGadgetsCrossingPaths,fig:validCover}, the \emph{outer
vertices} are the vertices $1$ and $2$, and the vertex $3$ is the only
\emph{inner vertex}. We now describe how \autoref{ob:nonzero,ob:zero} translate
into a system of equations, a solution of which is the gadget we
seek. Say one wants to represent a clause of degree $d$, and requires $k$ inner
degrees of freedom, \ie $k$ additional inner vertices that allow to have
\emph{enough} possible paths to fulfill all conditions. The gadget is a graph
with $k+d$ vertices, $d$ of which are outer vertices and $k$ are inner vertices.
The adjacency matrix of the graph is thus of order $k+d$.

\paragraph{Nonzero contribution. (\autoref{ob:nonzero})} We consider a cycle
cover consistent with an assignment of the variables. Recall that each boolean
variable $x_i$ is associated to a (blue) cycle in the graph, and that if that
a particular cycle is present in the cycle cover then this matches the variable
assignment $x_i = 0$ and if it is not then $x_i = 1$. For a gadget representing
a clause of degree $d$ we write $\s$ the bit-string representing one of the
$2^d$ possible assignments of the variables. The polynomial equation associated
to $\s$ is the permanent of the graph induced by all inner vertices and the
outer vertices that are indexed by 1 in $\s$, or equivalently the permanent of
the adjacency matrix where all rows and columns indexed by 0 in $\s$ are
removed. The polynomial must equate to $1$ whenever $f = 0$, and $e^{i \theta}$
for a constant $\theta \in \mathbb C$ otherwise. This naturally induces $2^d$
equations.

\paragraph{Zero contribution. (\autoref{ob:zero})} We consider now a cycle cover
inconsistent with an assignment of the variables. When looking only at a
traversed gadget, a \emph{partial cycle cover} is a cycle cover with a path entering the gadget
in vertex $x$ and exiting in vertex $y$ for $x$ different from $y$. We write
such a partial cycle cover $x \toc y$ and the set of all
cycle covers of this form $\mathscr C_{x,y}$. A path $x \to y$ (that can traverse any inner
vertices and possibly outer vertices) is called a \emph{crossing path} of the
partial cycle cover. \autoref{fig:quadraticGadgetsCrossingPaths} illustrates two
partial cycle covers of the form $1\toc 2$. The term \emph{partial cycle cover}
may be misleading; the partial cycle covers are indeed actual cycle covers, but they are
partial from the perspective of a single gadget, in the sense that they do not
cover the whole gadget. Incidentally, the cycle covers that correspond to a
valid assignment of the variables are those that are not \emph{partial} for any
gadget of the final graph. This comes from observing that, if a cover partially covers a gadget,
then the cycles representing the variable connected to the gadget must also be
partially covered (see \autoref{fig:quadraticGadgetsCrossingPaths}) and
consequently, that cover does not correspond to a valid assignment of the
variables. A cycle cover can been seen as a set of oriented pairs of edges representing the arcs (e.g. $\{\{1,1\},\{2,2\},\{3,3\}\}$ for \autoref{fig:coverCompleteGraphCover1}), hence inclusion is well
defined. For a partial cycle cover $x \toc y$, we denote by $\omega(x \toc y)$
the weight of the cycle cover, \ie the product of the weights of the edges
involved in it. From the above, having captured the notion of inconsistency by partial covers, we can now rephrase \autoref{ob:zero} as a more formal statement:

\begin{center}
     \textit{A graph gadget $A(V_A, E_A)$ is valid if the sum of the weights
of all its partial cycle covers is zero. }
\end{center}

\noindent Let us denote by $\OO_A \subseteq V_A$
the set of outer vertices of $A$. From the perspective of a whole graph $G(V_G,
E_G)$ (in which $A$ is \emph{included}, \ie $V_A \subset V_G$ and $E_A
\subset E_G$), the contribution of the partial cycle covers of $A$ is zero if
and only if
\begin{equation}\label{eq:zeroContribution}
        \sum_{
        \substack{(x, y) \in \OO_A \times \OO_A
            \\ x \neq y
        }}
            \left(
                \sum_{\tilde c \in \mathscr C_{(x,y)}} \omega(\tilde c) \per{\mathcal G_{\tilde c}}
            \right) = 0,
\end{equation}
where $G_{\tilde c}$ is some graph build upon the graph induced by the edges in
the cycle cover $\tilde c$. Here, the tilde is used to denote partial cycle covers.
Designing a valid gadget therefore implies to find the suitable weights of the
edges such that the above condition is satisfied. On the one hand, imposing
$\omega(\tilde c) = 0$ for all partial cycle cover $\tilde c$ of the gadget is
too restrictive, as it would imply that the gadget is not a connected graph, but
also that none of the nonzero contribution can be fulfilled. On the other hand,
we cannot consider all partial cycle covers at once and simply imposing
${\sum_{\tilde c}\omega(\tilde c) = 0}$, as this would not imply
\autoref{eq:zeroContribution} in the general case.

We denote by $\mathscr P(S)$
the power set of a set $S$, and write ${\mathscr V_{x,y} = \mathscr P(\OO_A -
\{\{x\}, \{y\}\})}$ the set of outer vertices a cycle cover of $\mathscr
C_{x,y}$ can traverse. Let ${\mathscr C_{x,y}^o \subset \mathscr C_{x,y}}$ denote
the set of partial cycle covers entering in $x$, exiting in $y$ and traversing a
set $o$ of outer vertices, so that 
\begin{equation}\label{eq:partialCycleCoverSet}
    \mathscr C_{x,y} = \bigcup_{o \in \mathscr V_{x,y}} \mathscr C_{x,y}^o.
\end{equation}

The derivation is illustrated via a series of illustrative examples in
\autoref{fig:generalizedCrossingPath,fig:G_oxy,fig:G_oxtoy}. The following
\autoref{fig:generalizedCrossingPath} illustrates the concept of partial cycle
covers on a gadget with three outer vertices, for a particular choice of pair
$(x, y)$.
\begin{figure}[th!]
    \centering
    \begin{subfigure}{.95\columnwidth}
        \centering
        \begin{tikzpicture}[
    shorten > = 1pt, 
    thin, 
    scale=1,
    every node/.append style={fill=white}
]

\node[rectangle,
rounded corners=0.05cm,
draw=violet,
fill = violet!20,
minimum width = 2.8cm,
minimum height = 2.2cm
] (r) at (.9, -.6) {};

\node[cloud,
draw = black,
max width = 1cm,
cloud puff arc = 180
] (c) at (-1.85,-.6) {\footnotesize $G_{\emptyset, (1,2)}$};

\node[cloud,
draw = black,
max width = 1cm,
cloud puff arc = 180
] (c) at (3.65, -.6) {\footnotesize $\quad G_{\overline \emptyset}\quad $};

\node[draw, circle, inner sep = 0.7ex, fill=violet, text=white] at (0, 0) (2) {2};
\node[draw, circle, inner sep = 0.7ex, fill=violet, text=white] at (0, -1.2) (1) {1};
\node[draw, circle, inner sep = 0.7ex, fill=violet, text=white] at (1.8, -.6) (3) {3};

\node[draw, circle, inner sep = 0.4ex] at (3, -.6) (3_) {};
\node[draw, circle, inner sep = 0.4ex] at (-1.35, 0) (2_) {};
\node[draw, circle, inner sep = 0.4ex] at (-1.35, -1.2) (1_) {};

\path[->, color=ultramarine, ultra thick] (3) edge [bend right=40]  (3_);
\path[->, color=ultramarine, ultra thick] (3_) edge [bend right=40]  (3);
\path[->, color=ultramarine, ultra thick] (1_) edge [bend right=40]  (1);

\path[->, color=ultramarine, ultra thick] (2) edge [bend right=40]  (2_);
\draw [ultra thick, decorate]
(1)  edge[->, bend right = 80] node[midway, fill=violet!20] {\small $\mathscr
C_{1,2}^\emptyset$} (2) ;


\end{tikzpicture}
        \caption{$o = \emptyset$}
        \label{fig:crossingEmptySet}
    \end{subfigure}

    \begin{subfigure}{.95\columnwidth}
        \centering
        \begin{tikzpicture}[
    shorten > = 1pt, 
    thin, 
    scale=1,
    every node/.append style={fill=white}
]

\node[rectangle,
rounded corners=0.05cm,
draw=violet,
fill = violet!20,
minimum width = 2.8cm,
minimum height = 2.2cm
] (r) at (.9, -.6) {};

\node[cloud,
draw = black,
max width = 1cm,
cloud puff arc = 180
] (c) at (-1.85,-.6) {\footnotesize $G_{\{3\}, (1,2)}$};

\node[cloud,
draw = black,
max width = 1cm,
cloud puff arc = 180
] (c) at (3.65, -.6) {\footnotesize $\ \ G_{\overline{\{3\}}}\ \ $};

\node[draw, circle, inner sep = 0.7ex, fill=violet, text=white] at (0, 0) (2) {2};
\node[draw, circle, inner sep = 0.7ex, fill=violet, text=white] at (0, -1.2) (1) {1};
\node[draw, circle, inner sep = 0.7ex, fill=violet, text=white] at (1.8, -.6) (3) {3};

\node[draw, circle, inner sep = 0.4ex] at (3, -.6) (3_) {};
\node[draw, circle, inner sep = 0.4ex] at (-1.35, 0) (2_) {};
\node[draw, circle, inner sep = 0.4ex] at (-1.35, -1.2) (1_) {};

\path[->, color=ultramarine, ultra thick] (1_) edge [bend right=40]  (1);

\path[->, color=ultramarine, ultra thick] (2) edge [bend right=40]  (2_);

\path [ultra thick]
    (1) edge[bend right=40] (3) (3) edge[->, bend right = 40] node[midway,
    fill=violet!20] {\small $\mathscr C_{1,2}^{\{3\}}$} (2);

\end{tikzpicture}
        \caption{$o = \{3\}$}
        \label{fig:crossingSetThree}
    \end{subfigure}
\caption{The purple rectangle represents a gadget with three outer vertices, the
inner vertices are not shown. For the choice of pair $(x,y) = (1,2)$,
\autoref{fig:crossingEmptySet} (resp. \autoref{fig:crossingSetThree})
illustrates sets of partial cycle covers for the choice of outer vertices  $o =
\emptyset$ (resp. $o = \{3\}$). 
}
\label{fig:generalizedCrossingPath}
\end{figure}
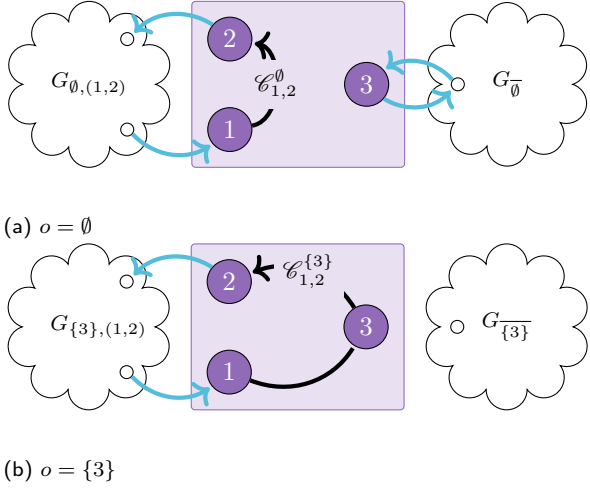

Let $c$ be a cycle cover of the full graph and ${\tilde c_o \in \mathscr
C_{x,y}^o},\ \tilde c_o \subset c$ be a partial cycle cover of the gadget. $G_c$
is, in all generality, composed of a component connected to $x$ and $y$ which we
write $G_{o,(x,y)}$ and a disconnected component $G_{\overline{o}}$ (which is
connected to all outer vertices not in $o$, thus disconnected from the partial
cycle cover). $G_{o,(x,y)}$ indeed contains the edges of $\tilde c_o$ (see
\autoref{fig:G_oxy}).
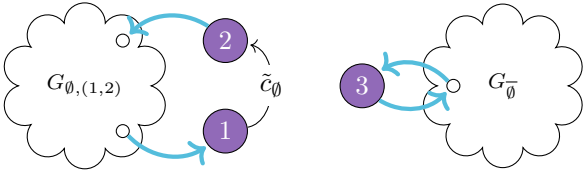
\begin{figure}[h!t]
        \centering
        \begin{tikzpicture}[
    shorten > = 1pt, 
    thin, 
    scale=1,
    every node/.append style={fill=white}
]


\node[cloud,
draw = black,
max width = 1cm,
cloud puff arc = 180
] (c) at (-1.85,-.6) {\footnotesize $G_{\emptyset, (1,2)}$};

\node[cloud,
draw = black,
max width = 1cm,
cloud puff arc = 180
] (c) at (3.65, -.6) {\footnotesize $\quad G_{\overline \emptyset} \quad$};

\node[draw, circle, inner sep = 0.7ex, fill=violet, text=white] at (0, 0) (2) {2};
\node[draw, circle, inner sep = 0.7ex, fill=violet, text=white] at (0, -1.2) (1) {1};
\node[draw, circle, inner sep = 0.7ex, fill=violet, text=white] at (1.8, -.6) (3) {3};

\node[draw, circle, inner sep = 0.4ex] at (3, -.6) (3_) {};
\node[draw, circle, inner sep = 0.4ex] at (-1.35, 0) (2_) {};
\node[draw, circle, inner sep = 0.4ex] at (-1.35, -1.2) (1_) {};

\path[->, color=ultramarine, ultra thick] (3) edge [bend right=40]  (3_);
\path[->, color=ultramarine, ultra thick] (3_) edge [bend right=40]  (3);
\path[->, color=ultramarine, ultra thick] (1_) edge [bend right=40]  (1);

\path[->, color=ultramarine, ultra thick] (2) edge [bend right=40]  (2_);
\draw [decorate]
(1)  edge[->, bend right = 80] node[midway] {$\tilde c_{\emptyset}$} (2) ;


\end{tikzpicture}
        \caption{Following the example of \autoref{fig:generalizedCrossingPath},
        we set ${o = \emptyset}$, which in turn sets $G_{\emptyset, (1,2)}$ and
        $G_{\overline \emptyset}$. Then, we choose a particular partial cycle cover
        $\tilde c_\emptyset \in \mathscr C_{1,2}$.}
        \label{fig:G_oxy}
\end{figure}

We now describe how to construct a graph $G_{o, x \to y}$ from $G_{o,(x,y)}$
which will be such that ${\per{\mathcal G_{o,(x,y)}} = \per{\mathcal G_{o, x \to y}} \omega(\tilde
c_o)}$. $G_{o, x \to y}$ is built from the component connected to $x$ and $y$ by
removing the edges going from $x$ and $y$ to any other vertex in $A$ and adding
the edge $x \to y$ of weight one (see \autoref{fig:G_oxtoy}). 

\begin{figure}[h!t]
        \centering
        \begin{tikzpicture}[
    shorten > = 1pt, 
    thin, 
    scale=1,
    every node/.append style={fill=white}
]


\node[cloud,
draw = black,
max width = 1cm,
cloud puff arc = 180
] (c) at (-1.85,-.6) {\footnotesize $G_{\emptyset, 1\to 2}$};

\node[cloud,
draw = black,
max width = 1cm,
cloud puff arc = 180
] (c) at (3.65, -.6) {\footnotesize $\quad G_{\overline \emptyset} \quad$};

\node[draw, circle, inner sep = 0.7ex, fill=violet, text=white] at (0, 0) (2) {2};
\node[draw, circle, inner sep = 0.7ex, fill=violet, text=white] at (0, -1.2) (1) {1};
\node[draw, circle, inner sep = 0.7ex, fill=violet, text=white] at (1.8, -.6) (3) {3};

\node[draw, circle, inner sep = 0.4ex] at (3, -.6) (3_) {};
\node[draw, circle, inner sep = 0.4ex] at (-1.35, 0) (2_) {};
\node[draw, circle, inner sep = 0.4ex] at (-1.35, -1.2) (1_) {};

\path[->, color=ultramarine, ultra thick] (3) edge [bend right=40]  (3_);
\path[->, color=ultramarine, ultra thick] (3_) edge [bend right=40]  (3);
\path[->, color=ultramarine, ultra thick] (1_) edge [bend right=40]  (1);

\path[->, color=ultramarine, ultra thick] (2) edge [bend right=40]  (2_);
\draw [decorate]
(1)  edge[->] (2) ;


\end{tikzpicture}
        \caption{Continuing with the example initiated in \autoref{fig:G_oxy},
        we build $G_{\emptyset, 1 \to 2}$ from $G_{\emptyset, (1,2)}$ by
        removing the edges arriving in vertices $1$ and $2$ to any vertex of the
        gadget, and adding the edge $1 \to 2$ of weight 1.}
    \label{fig:G_oxtoy}
\end{figure}
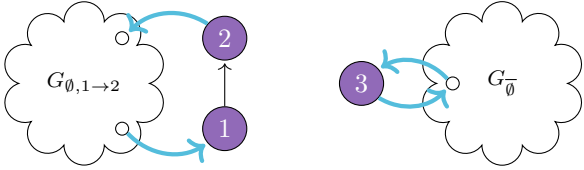
Overall, the weight of the cycle cover $c$ is ${\per{\mathcal G_{o, x\to y}}
\per{\mathcal G_{\overline{o}}} \omega(\tilde c_o)}$. We can observe that ${\per{\mathcal G_{o, x\to y}}\per{\mathcal G_{\overline{o}}}}$
remains unchanged for any choice of $\tilde c_o$ for fixed $o$. In that regard,
\autoref{eq:zeroContribution} rewrites
\begin{equation}\label{eq:traversedGadgetContribution}
        \begin{aligned}
        \sum_{\substack{(x, y) \in \OO_A^2 \\ x \neq y}}
                \sum_{o \in \mathscr V_{x,y}} 
                \Big( & \per{\mathcal G_{o,x\to y}}
                \per{\mathcal G_{\overline{o}}} \\
                & \times 
                    \sum_{
                        \tilde c_o \in \mathscr C_{x,y}^o
                    }
                    \omega(\tilde c_o) \Big)
            = 0.
            \end{aligned}
\end{equation}

While designing a gadget with fixed number of inner vertices, the only control one has is on the weights
$\omega(\tilde c_o)$, as all the other terms of
\autoref{eq:traversedGadgetContribution} depend on a specific instance. But
still, one must ensure that \autoref{eq:traversedGadgetContribution} holds.
Hence, in order to design a clause gadget whose supposedly zero contributions
are effectively zero, one must ensure that, for all $(x, y) \in \OO_A \times
\OO_A$ with $x \neq y$ and for all $o \in \mathscr P (\OO_A - \{\{x\}, \{y\}\} )$, the following
holds

\begin{equation}\label{eq:isZeroAndMustBeZero}
    \sum_{
        \tilde c_o \in \mathscr C_{x,y}^o
    } \omega(\tilde c_o) = 0 .
\end{equation}
Applying \autoref{eq:isZeroAndMustBeZero} to the example of
\autoref{fig:generalizedCrossingPath} for the fixed pair $(1, 2)$ yields:
\begin{equation}
        \sum_{
            \tilde c_\emptyset \in \mathscr C_{1,2}^\emptyset
        } \omega(\tilde c_\emptyset) 
        = \sum_{
            \tilde c_{\{3\}}' \in \mathscr C_{1,2}^{\{3\}}
            } \omega(\tilde c_{\{3\}}') 
        = 0.
\end{equation}
Consequently, a gadget $A$ might require several inner vertices to have enough
degrees of freedom (the supplementary nodes and the weights on their connected
edges) to satisfy \autoref{eq:isZeroAndMustBeZero}. Note that it does not show
that grouping the partial cycle cover with respect to the set of outer vertices
they traverse is the most efficient way to design a gadget, in the sense that
there might be a way to group them in bigger groups, therefore lowering the
number of inner vertices required. However,
\autoref{eq:traversedGadgetContribution} seems to indicate that this choice is
optimal as we only have a local control on each contribution (the $\omega(\tilde
c)$'s).

\begin{algorithm}[t!]
    \caption{Procedure for the generation of the clause gadget constraints}
    \SetKwInOut{Input}{inputs}
    \SetKwInOut{Output}{output}
    \Input{\begin{itemize}[label=-]
        \setlength\itemsep{.1em}
        \item $d$ the degree of the clause,
        \item $k$ the number of inner degrees of freedom,
        \item $m$ a clause of degree $d$.
    \end{itemize}}
    \Output{The family of polynomial equations associated to the input parameters.}
    \Begin{
     Let $F$ be the list of polynomials.\\
     Let $A(V_A, E_A)$ be the symbolic gadget with adjacency matrix $\mathcal A :=
     (x_{ij})_{1 \leq i, j \leq k+d}$.\\
     Let $\OO$ (resp. $\II$) be the set of outer (resp. inner) vertices indexes
     of $A(V_A, E_A)$ (\ie $V_A = \OO \cup \II$). \\

    \For{\emph{\textbf{all}} $\s \in \{0, 1\}^d$}{
         $\OO_{\s} \leftarrow \left\{i\ |\ i \in \OO,\, s_i = 1 \right\}$\\
         $A'(V'_A,E'_A) \leftarrow A'(\OO_{\s}\cup \II, E_{A'})$ \Comment*{$E_{A'}$ is the set of
         edges of $A$ induced by the vertices in $\OO_{\s}\cup \II$}
         $F \leftarrow  F \cup \per{\mathcal A'}\ -\ e^{i m}$ }
    \For{\emph{\textbf{all}} $(x, y) \in \OO \times \OO$ s.t. $x \neq y$}{
        \For {$\bs v \in \mathscr V_{x,y}$}
        {
        $A''(V''_A,E''_A) \leftarrow A''(\{x, y\} \cup \bs v \cup \II, E_{A''})$ \\
        Set the row representing $y$ to zero, apart from the $x$-th entry that
        is set to 1 in $A''$.\\
        $F \leftarrow F \cup \per{\mathcal A''}$
        }
    }
    \Return $F$
    }
    \label{alg:generation}
\end{algorithm}

The construction of the system of equations for the zero contributions is
slightly more involved than the nonzero case. Given two outer vertices $x$ and
$y$ in $\OO_A$, we want that all the partial cycle covers of $\mathscr C_{x,y}$
have zero contribution. In order to write the constraints on the contributions
of all partial covers $\mathscr C_{x,y}$  -- comprising the paths $x \to y$ --
as permanent equations we use the following trick. Given $x$ and $y$, we set to
0 all the entries of the $y$-th row of the adjacency matrix, apart from the
$x$-th entry that we set to 1. This way, all edges exiting $y$ are removed and
an artificial edge $y \to x$ is created, so that all crossing paths $x \to y$
are transformed into cycles. As a result, the permanent of the remaining graph
is exactly the sum of all the contributions of the partial covers of $\mathscr
C_{x,y}$. Recall that we grouped the covers in $\mathscr C_{x,y}$ according to
the set of outer vertices they traverse (see \autoref{eq:partialCycleCoverSet}),
hence this procedure is applied \emph{separately} to such set. From a pair $(x,
y)$ and a set of outer vertices, the procedure generates a single polynomial. We
repeat this procedure for all pairs of outer vertices, and we obtain a system of
$2^{d-2}d(d-1)$ equations. The overall procedure to generate the system of
equations is summarized in \autoref{alg:generation}, and the resulting system of
polynomial equations is characterized in the following
\autoref{pr:characterization}. Observe that the algorithm does no computation,
\ie no permanent is computed, but the variables are symbolically manipulated to
write down polynomials. \autoref{pr:gadgetExistence} asserts that whether a
desired design is possible can be decided by computing the reduced Gröbner basis
of the system of polynomials.

\begin{restatable}[Gadget existence]{thm}{gadget}\label{pr:gadgetExistence}
    There exists a gadget encoding the clause 
    \begin{equation}
        m=\theta \prod_{\substack{i \in S \\ S  \subseteq [n] \\ |S| = d}} x_i
    \end{equation} with $d+k$ vertices if and only if the
    system of polynomial equations created by \autoref{alg:generation} on
    parameters $d$, $k$ and $m$ has a reduced Gröbner basis different from
    $\{1\}$.
\end{restatable}

\begin{proof}
    The proof follows from Hilbert's Nullstellensatz \cite[see Chapter
    4]{cox_ideals_2013} and the correctness of \autoref{alg:generation}.
    
    \emph{Proof of correctness of \autoref{alg:generation}.} The first \emph{for
    all} loop writes all the constraints related to \autoref{ob:nonzero}. The
    second \emph{for all} loop tweaks the graph to match
    \autoref{eq:isZeroAndMustBeZero} and write the constraints related to
    \autoref{ob:zero}.
    
    \emph{Proof of gadget existence.} Let $F$ be the system of polynomial
    equations generated by \autoref{alg:generation} on parameters $d$, $k$ and
    $m$. If the reduced Gröbner basis of $\braket{F}$ is not $\{1\}$, then $F$
    has (at least) a solution. 
\end{proof}
Besides its existence, the construction of such a gadget is ensured by the
elimination theory \cite[see Chapter 3]{cox_ideals_2013}, provided the reduced
Gröbner basis is computed with respect to an elimination order. A point $\x =
(x_1,\cdots, x_{(k+d)^2})$ such that $F(\x) = 0$  corresponds to the entries of
the adjacency matrix of the gadget. Incidentally, we have the following
\autoref{eq:optimalK}.

\begin{corollary}[Smallest graph gadget]\label{eq:optimalK}
    The optimal value of $k$ to design a gadget is the smallest integer such
    that the reduced Gröbner basis of the system of polynomial equations
    generated by \autoref{alg:generation} on parameters $d$, $k$ and $m$
    is not $\{1\}$. 
\end{corollary}

\begin{proposition}[Characterisation of the system of equations]
    \label{pr:characterization}
    Let $A$ be the adjacency matrix of a gadget representing a clause of degree
    $d$ with $k$ inner degrees of freedom. Then, the \autoref{alg:generation}
    generates a system of $\ 2^d\left(d(d-1)2^{-2} + 1\right)$ multilinear
    polynomial equations in $(k+d)^2$ variables. It runs in time
    $\Theta(2^d)$.
\end{proposition}
\begin{proof}
    Each entry of the $(k+d) \times (k+d)$ adjacency matrix of the gadget is a
    variable hence the number of variables in the polynomial system is
    $(k+d)^2$. The first \emph{for all} loop of \autoref{alg:generation}
    generates $2^d$ equations, and the second generates $2^{d-2}d(d-1)$.
\end{proof}

A drawback of \autoref{alg:generation} is that the set of polynomial $F$
contains a lot of redundancy. Therefore, in practice before computing the reduced Gröbner
basis of a system of polynomial equations, we advise using Laplace's expansion
formula along a row, which reads for the $i$-th
\begin{equation}
\per{A} = \sum_{j=1}^n a_{i,j} \per{A_{\bar{i},\bar{j}}},
\end{equation}
where $A_{\bar{i},\bar{j}}$ is the submatrix of $A$ obtained by removing the
$i$-th row and the $j$-th column. This relies on the fact that
\autoref{alg:generation} sets the value of many submatrix permanents
\smash{$\per{A_{\bar{i},\bar{j}}}$}, which can then be substituted in the rest
of the system. If a numerical solution is sought, one can use homotopy
continuation \cite{chen_homotopy_2015}, as many computer algebra system do not
implement solution finding from the reduced Gröbner basis over the complex
numbers when the variety is not zero-dimensional \cite{berthomieu_msolve_2021,
demin_groebner_2023}. Finding a solution, that is, a clause gadget, takes
exponential time \cite{bardet_complexity_2015}, but is an
operation that only has to be done once for each degree. Also, it is to be noted
that both polynomial time encodings of
\cite{valiant_complexity_1979,rudolph_simple_2009} do not take into account the
computational complexity of the task of finding a gadget.

\subsection{Computation of some graph gadgets and application to quantum
circuits amplitude computation }\label{sec:example} 

In this section, we give an illustrative example by deriving the graph gadgets
for linear, quadratic and cubic clauses, and show how this can be applied to
quantum circuit probability amplitudes computation.

\paragraph{Gadget computation.} We denote by $A_k(\theta)$ the adjacency matrix
of the gadget encoding a clause of degree $k$ with coefficient $\theta$. ITs and  $k$
first rows/columns index the outer vertices of the gadget, while the remaining
ones indicate the inner vertices. The following gadgets are derived by solving
the system of equations returned by \autoref{alg:generation} on the appropriate
parameters. The computations were performed using the Julia programming language
\cite{demin_groebner_2023,bezanson_julia_2017}. The degree-one monomials with
coefficient $\theta$ are encoded into the following gadget represented by a $1
\times 1$ matrix:
\begin{equation*}
A_1(\theta) := \begin{bmatrix}
    e^{i\theta}
\end{bmatrix}.
\end{equation*}

The quadratic clauses with coefficient $\theta$ are encoded into a gadget with
one inner vertex:
\begin{equation*}
A_2(\theta) :=\begin{bmatrix}\frac{1}{2}(1+e^{i \theta}) & \frac{1}{2}(e^{i
    \theta} - 1) & \frac{1}{2}(1 - e^{i \theta})\\-1 & 0 & 1\\1 & 1 &
    1\end{bmatrix}.
\end{equation*}

Finally, the cubic clauses with coefficient are encoded into the following
gadget comprising two inner vertices. Let $\eta(\theta) := \frac{1}{6}\sqrt{3
(1+i)(1 - e^{i \theta})}$,
\begin{equation*}
A_3(\theta) :=
\begingroup 
\setlength\arraycolsep{1pt}
\begin{bmatrix}
    \frac{e^{i \theta} - (1 + 12 i)}{-12 i}
    & -\eta(\theta)
    & -\eta(\theta)
    &  \frac{\eta(\theta)}{\sqrt{2}} e^{i \frac \pi 4}
    &  \frac{-\eta(\theta)}{\sqrt{2}}
\\  -\eta(\theta)
    & i
    & -1 +i
    & 1
    & e^{i \frac{3\pi}{4}}
\\ -\eta(\theta)
    & -1 +i
    &i
    & 1
    & e^{i \frac{3\pi}{4}}
\\ \frac{\eta(\theta)}{\sqrt{2}} e^{i \frac \pi 4}
    & 1
    & 1
    & 1
    & 0
\\- \sqrt{ \frac{\left(1 -e^{i \theta}\right)}{24(1+i)}}
    & e^{i \frac{3\pi}{4}}
    & e^{i \frac{3\pi}{4}}
    & 0
    & 1
\end{bmatrix}.
\endgroup
\end{equation*}

Foremost, for each gadget the number of inner vertices is minimal according to
\autoref{eq:optimalK}.

\paragraph{Quantum circuits.} These gadgets can be used to encode probability
amplitudes of a broad family of quantum circuits, we exemplify with an IQP
circuit. A simple gate-set, which comes with hardness results regarding
classical simulation is that of Instantaneous Quantum Polynomial-time (IQP)
circuits
\cite{shepherd_temporally_2009,bremner_classical_2011,bremner_average-case_2016}.
IQP circuits are circuits $\CC = H^{\otimes q} \mathcal D H^{\otimes q}$, with
$\mathcal D$ diagonal in the $Z$-basis. They were introduced as a subset of
quantum circuits that is rich enough to enable sampling from distributions that
is hard to describe classically.

We shall call a \emph{balanced} gate any gate $\UU(\theta)$ whose amplitudes are
of the form
\begin{equation}\label{eq:amplitudeBalancedGate}
    \braket{\y | \UU(\theta) | \x} =
\begin{cases}
  0 \text{ or } \\
  \frac{e^{i \theta f(\x, \y)}}{\sqrt M} ,
\end{cases}
\end{equation}
where $M$ is the number of states $\ket{\y}$ such that $\braket{\y | \UU | \x}$
is nonzero and $f \in \mathbb F_2[\x, \y]$ is a polynomial. Here,
\emph{balanced} indicates that all amplitudes have the same magnitude.
 Examples
of gates satisfying \autoref{eq:amplitudeBalancedGate} are the (arbitrary
controlled) phase gates ($M = 1$), defined as
\begin{equation}\label{eq:defPhaseGates}
        P_k(\theta) : \ket{x_1, \cdots, x_k} \mapsto e^{i \theta \Pi_{j=1}^kx_j} \ket{x_1, \cdots, x_k},
\end{equation}
for all $k \geq 1$ and $x_j \in \{0,1\}$ for all $1 \leq j \leq k$, or the
Hadamard gates ($\theta = \pi$, $M=2$). Suitable gate-sets for $\mathcal D$
could be $\{P_0(\frac{\pi}{8}), P_1(\frac{\pi}{4})\}$ or $\{P_k(\pi)|,k =
0,1,2\}$ \cite{bremner_average-case_2016}. 
The polynomial $f$ is a function of the computational basis bit-strings, in
which each variable is either raised to the zeroth or the first power, \ie a
clause of degree $d$ contains $d$ distinct variables -- so that the resulting
polynomial is a multilinear polynomial of degree $d$. Probability amplitudes of
quantum circuits built upon balanced gates can be straightforwardly encoded onto
a graph. Using encoding of quantum circuit amplitudes onto polynomials
\cite{dawson_quantum_2004,montanaro_quantum_2017} and the gadgets we derived,
each $P_k(\theta)$ gate is to be encoded into a gadget $A_k(\theta)$. The
resulting graph is a graph $G_\CC(V,E)$ such that
\begin{equation}
    \braket{\a | \mathcal C | \b } = \frac{1}{\sqrt{2}^h} \sum_{\x} e^{i f(\x)} = \frac{1}{\sqrt{2}^h} \per{\GG_\CC}.
\end{equation}

Remark that $A_2(\pi)$ is the gadget encoding the quadratic clauses of
\cite{rudolph_simple_2009}, which is consistent with the observation that both
Hadamard (regardless its rescaling factor) and CZ amplitudes are of the form
$(-1)^{xy}$ for the suitable choice of variables $x$ and $y$. The encoding of
\cite{rudolph_simple_2009} therefore comes as a special case of our method.

With these three gadgets in hand, the graph encoding of $q$-qubit IQP circuits
is characterized by the following \autoref{pr:numberOfNodes}. Throughout rest,
let $\#(\UU)$ denote the number of occurrences of the gate $\UU$ in the circuit
$\CC$, \eg $\#(H) = 2q$.

\begin{proposition}[Size of the resulting graph]\label{pr:numberOfNodes} Let $\CC$ be a $q$-qubit IQP circuit
    of depth $d$ and $G_\CC(V,E)$ be the graph such that
    ${\braket{0^{\otimes{q}} | \CC | 0^{\otimes{q}}} =
    \frac{1}{2^q}\per{\GG_\CC}}$. Then, $G_\CC$ is a graph with 
    \begin{equation} M = q + \#(P_0(\theta)) + 3\#(P_1(\theta)) + 5\#(P_2(\theta))
    \end{equation}
    vertices.
\end{proposition}
\begin{proof}
    The proof simply follows from counting the number of nodes of each gadget.
\end{proof}

\paragraph*{Example.} For the example, we consider the 3-qubit circuit depicted
in \autoref{fig:exampleIQP} which we call $\CC$.
\begin{figure}[hbtp!]
    \centering
    \input{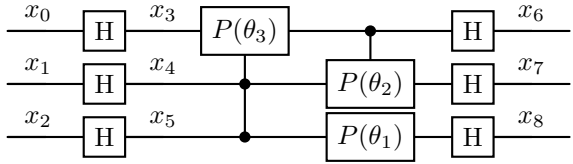}
    \caption{Example 3-qubit circuit with its associated variables $x_0, \cdots, x_8$.}
    \label{fig:exampleIQP}
\end{figure} 
Following the method for creating the variables on the circuit wires, we have
$9$ variables. Write ${\x = (x_0, \ldots, x_8)}$, the polynomial associated to
$\CC$ is
\begin{equation}
    \begin{aligned}
        \hspace{-.8em}f(\x) =  
         \hspace{-0.1cm}\sum_{i=0}^5 \pi  x_i x_{i+3} 
           + \theta_3x_3x_4x_5 + \theta_2x_3x_4 + \theta_1x_5.
    \end{aligned}
\end{equation}
and the amplitude as function of the variables reads
\begin{equation}
\braket{x_6, x_7, x_8 |\ \CC\ | x_0, x_1, x_2} = \frac{1}{2^3} \sum_{\x}e^{i f(\x)}.
\end{equation}
For the sake of the example, let's compute the amplitude $\braket{0^{\otimes 3}
| \CC |0^{\otimes 3}}$ via the graph encoding, \ie set all the variables but
$x_3, x_4, x_5$ to zero. Also, set arbitrarily $\theta_1 := \frac{\pi}{2}$,
$\theta_2 := \frac{\pi}{4}$ and $\theta_3 := \frac{\pi}{8}$. Then, construct the
graph from the gadgets $A_i(\theta_i)$, $i = 0, 1, 2$; it is shown in
\autoref{fig:graphExampleIQP}. Note that since the variables encoding the
amplitude are set to zero, a slightly more economical encoding is possible for
all the clauses involving theses variables as there is no polynomial to evaluate
(see \autoref{eq:amplitudeBalancedGate}): it is the constant polynomial $0$.
Finally, we compute the permanent of $\GG_\CC$ to get the amplitude. Numerical
computation yields
\begin{equation}
    \braket{0^{\otimes 3} | \CC | 0^{\otimes 3}} = \frac{1}{8} \per{\GG_\CC} \approx (0.348+0.511i).
\end{equation}

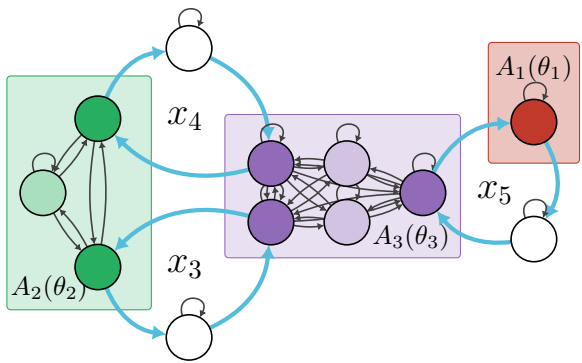
\begin{figure}[hbtp!]
    \centering
\begin{tikzpicture}[scale=1.7] 
\begin{scope}[on background layer]
    \node[rectangle,
        rounded corners=0.05cm,
        draw=violet,
        fill = violet!20,
    	minimum width = 3.1cm, 
    	minimum height = 1.9cm 
     ] (r) at (2.92,3) {};

     \node[rectangle,
        rounded corners=0.05cm,
        draw=darkgreen,
        fill = darkgreen!20,
    	minimum width = 1.9cm, 
    	minimum height = 3.1cm 
     ] (r) at (0.88,2.9475) {};

      \node[rectangle,
        rounded corners=0.05cm,
        draw=darkred,
        fill = darkred!30,
    	minimum width = 1.2cm,
    	minimum height = 1.6cm
     ] (r) at (4.400, 3.65) {};

\end{scope}
\node[] () at (4.400, 3.92) {$A_1(\theta_1)$};
\node[] () at (.65,2.2) {$A_2(\theta_2)$};
\node[] () at (3.45,2.6) {$A_3(\theta_3)$};
\node[] () at (4.1, 2.9475)  {\Large$x_5$};
\node[] () at (1.7, 3.526) {\Large$x_4$};
\node[] () at (1.7, 2.369)  {\Large$x_3$};

\Vertex[x=1.731,y=4.075,size=0.2,color=white,style={semithick, minimum width=.6cm }]{0}
\Vertex[x=1.731,y=1.820,size=0.2,color=white,style={semithick, minimum width=.6cm }]{1}
\Vertex[x=4.4,y=2.587,size=0.2,color=white,style={semithick, minimum width=.6cm }]{2}
\Vertex[x=2.365,y=3.176,size=0.2,color=violet,style={semithick, minimum width=.6cm }]{3}
\Vertex[x=2.365,y=2.719,size=0.2,color=violet,style={semithick, minimum width=.6cm }]{4}
\Vertex[x=3.541,y=2.9475,size=0.2,color=violet,style={semithick, minimum width=.6cm }]{5}
\Vertex[x=2.953,y=2.719,size=0.2,color=violet!40,style={semithick, minimum width=.6cm }]{6}
\Vertex[x=2.953,y=3.176,size=0.2,color=violet!40,style={semithick, minimum width=.6cm }]{7}
\Vertex[x=1.025,y=3.526,size=0.2,color=darkgreen,style={semithick, minimum width=.6cm }]{8}
\Vertex[x=1.025,y=2.369,size=0.2,color=darkgreen,style={semithick, minimum width=.6cm }]{9}
\Vertex[x=0.600,y=2.9475,size=0.2,color=darkgreen!40,style={semithick, minimum width=.6cm }]{10}
\Vertex[x=4.400,y=3.5,size=0.2,color=darkred,style={semithick, minimum width=.6cm }]{13}
\Edge[,lw=0.2,bend=-8.531,style={semithick, -{Latex[length=0.09cm,width=0.1cm]},  },loopsize=0.2cm,loopposition=90,loopshape=45,Direct](0)(0)
\Edge[,lw=0.2,bend=30,style={ultra thick, ultramarine, -{Latex[length=0.18cm,width=0.2cm]},  },loopsize=0.2cm,loopposition=90,loopshape=45,Direct](0)(3)
\Edge[,lw=0.2,bend=-30,style={ultra thick, ultramarine, -{Latex[length=0.18cm,width=0.2cm]},  },loopsize=0.2cm,loopposition=90,loopshape=45,Direct](1)(4)
\Edge[,lw=0.2,bend=-8.531,style={semithick, -{Latex[length=0.09cm,width=0.1cm]},  },loopsize=0.2cm,loopposition=85,loopshape=45,Direct](1)(1)
\Edge[,lw=0.2,bend=30,style={ultra thick, ultramarine, -{Latex[length=0.18cm,width=0.2cm]},  },loopsize=0.2cm,loopposition=90,loopshape=45,Direct](2)(5)
\Edge[,lw=0.2,bend=-8.531,style={semithick, -{Latex[length=0.09cm,width=0.1cm]},  },loopsize=0.2cm,loopposition=90,loopshape=45,Direct](3)(3)
\Edge[,lw=0.2,bend=-8.531,style={semithick, -{Latex[length=0.09cm,width=0.1cm]},  },loopsize=0.2cm,loopposition=90,loopshape=45,Direct](3)(4)
\Edge[,lw=0.2,bend=-8.531,style={semithick, -{Latex[length=0.09cm,width=0.1cm]},  },loopsize=0.2cm,loopposition=90,loopshape=45,Direct](3)(5)
\Edge[,lw=0.2,bend=-8.531,style={semithick, -{Latex[length=0.09cm,width=0.1cm]},  },loopsize=0.2cm,loopposition=90,loopshape=45,Direct](3)(6)
\Edge[,lw=0.2,bend=-8.531,style={semithick, -{Latex[length=0.09cm,width=0.1cm]},  },loopsize=0.2cm,loopposition=90,loopshape=45,Direct](3)(7)
\Edge[,lw=0.2,bend=30,style={ultra thick, ultramarine, -{Latex[length=0.18cm,width=0.2cm]},  },loopsize=0.2cm,loopposition=90,loopshape=45,Direct](3)(8)
\Edge[,lw=0.2,bend=-8.531,style={semithick, -{Latex[length=0.09cm,width=0.1cm]},  },loopsize=0.2cm,loopposition=90,loopshape=45,Direct](4)(3)
\Edge[,lw=0.2,bend=-8.531,style={semithick, -{Latex[length=0.09cm,width=0.1cm]},  },loopsize=0.2cm,loopposition=90,loopshape=45,Direct](4)(4)
\Edge[,lw=0.2,bend=-8.531,style={semithick, -{Latex[length=0.09cm,width=0.1cm]},  },loopsize=0.2cm,loopposition=90,loopshape=45,Direct](4)(5)
\Edge[,lw=0.2,bend=-8.531,style={semithick, -{Latex[length=0.09cm,width=0.1cm]},  },loopsize=0.2cm,loopposition=90,loopshape=45,Direct](4)(6)
\Edge[,lw=0.2,bend=-8.531,style={semithick, -{Latex[length=0.09cm,width=0.1cm]},  },loopsize=0.2cm,loopposition=90,loopshape=45,Direct](4)(7)
\Edge[,lw=0.2,bend=-30,style={ultra thick, ultramarine, -{Latex[length=0.18cm,width=0.2cm]},  },loopsize=0.2cm,loopposition=90,loopshape=45,Direct](4)(9)
\Edge[,lw=0.2,bend=-8.531,style={semithick, -{Latex[length=0.09cm,width=0.1cm]},  },loopsize=0.2cm,loopposition=90,loopshape=45,Direct](5)(3)
\Edge[,lw=0.2,bend=-8.531,style={semithick, -{Latex[length=0.09cm,width=0.1cm]},  },loopsize=0.2cm,loopposition=90,loopshape=45,Direct](5)(4)
\Edge[,lw=0.2,bend=-8.531,style={semithick, -{Latex[length=0.09cm,width=0.1cm]},  },loopsize=0.2cm,loopposition=90,loopshape=45,Direct](5)(5)
\Edge[,lw=0.2,bend=-8.531,style={semithick, -{Latex[length=0.09cm,width=0.1cm]},  },loopsize=0.2cm,loopposition=90,loopshape=45,Direct](5)(6)
\Edge[,lw=0.2,bend=-8.531,style={semithick, -{Latex[length=0.09cm,width=0.1cm]},  },loopsize=0.2cm,loopposition=90,loopshape=45,Direct](5)(7)
\Edge[,lw=0.2,bend=30,style={ultra thick, ultramarine, -{Latex[length=0.18cm,width=0.2cm]},  },loopsize=0.2cm,loopposition=90,loopshape=45,Direct](5)(13)
\Edge[,lw=0.2,bend=-8.531,style={semithick, -{Latex[length=0.09cm,width=0.1cm]},  },loopsize=0.2cm,loopposition=90,loopshape=45,Direct](6)(3)
\Edge[,lw=0.2,bend=-8.531,style={semithick, -{Latex[length=0.09cm,width=0.1cm]},  },loopsize=0.2cm,loopposition=90,loopshape=45,Direct](6)(4)
\Edge[,lw=0.2,bend=-8.531,style={semithick, -{Latex[length=0.09cm,width=0.1cm]},  },loopsize=0.2cm,loopposition=90,loopshape=45,Direct](6)(5)
\Edge[,lw=0.2,bend=-8.531,style={semithick, -{Latex[length=0.09cm,width=0.1cm]},  },loopsize=0.2cm,loopposition=90,loopshape=45,Direct](6)(6)
\Edge[,lw=0.2,bend=-8.531,style={semithick, -{Latex[length=0.09cm,width=0.1cm]},  },loopsize=0.2cm,loopposition=90,loopshape=45,Direct](7)(3)
\Edge[,lw=0.2,bend=-8.531,style={semithick, -{Latex[length=0.09cm,width=0.1cm]},  },loopsize=0.2cm,loopposition=90,loopshape=45,Direct](7)(4)
\Edge[,lw=0.2,bend=-8.531,style={semithick, -{Latex[length=0.09cm,width=0.1cm]},  },loopsize=0.2cm,loopposition=90,loopshape=45,Direct](7)(5)
\Edge[,lw=0.2,bend=-8.531,style={semithick, -{Latex[length=0.09cm,width=0.1cm]},  },loopsize=0.2cm,loopposition=90,loopshape=45,Direct](7)(7)
\Edge[,lw=0.2,bend=-8.531,style={semithick, -{Latex[length=0.09cm,width=0.1cm]},  },loopsize=0.2cm,loopposition=90,loopshape=45,Direct](8)(9)
\Edge[,lw=0.2,bend=-8.531,style={semithick, -{Latex[length=0.09cm,width=0.1cm]},  },loopsize=0.2cm,loopposition=90,loopshape=45,Direct](8)(10)
\Edge[,lw=0.2,bend=30,style={ultra thick, ultramarine, -{Latex[length=0.18cm,width=0.2cm]},  },loopsize=0.2cm,loopposition=90,loopshape=45,Direct](8)(0)
\Edge[,lw=0.2,bend=-8.531,style={semithick, -{Latex[length=0.09cm,width=0.1cm]},  },loopsize=0.2cm,loopposition=90,loopshape=45,Direct](9)(8)
\Edge[,lw=0.2,bend=-8.531,style={semithick, -{Latex[length=0.09cm,width=0.1cm]},  },loopsize=0.2cm,loopposition=90,loopshape=45,Direct](9)(10)
\Edge[,lw=0.2,bend=-30,style={ultra thick, ultramarine, -{Latex[length=0.18cm,width=0.2cm]},  },loopsize=0.2cm,loopposition=90,loopshape=45,Direct](9)(1)
\Edge[,lw=0.2,bend=-8.531,style={semithick, -{Latex[length=0.09cm,width=0.1cm]},  },loopsize=0.2cm,loopposition=90,loopshape=45,Direct](10)(8)
\Edge[,lw=0.2,bend=-8.531,style={semithick, -{Latex[length=0.09cm,width=0.1cm]},  },loopsize=0.2cm,loopposition=90,loopshape=45,Direct](10)(9)
\Edge[,lw=0.2,bend=-8.531,style={semithick, -{Latex[length=0.09cm,width=0.1cm]},  },loopsize=0.2cm,loopposition=90,loopshape=45,Direct](10)(10)
\Edge[,lw=0.2,bend=-8.531,style={semithick, -{Latex[length=0.09cm,width=0.1cm]},  },loopsize=0.2cm,loopposition=90,loopshape=45,Direct](13)(13)
\Edge[,lw=0.2,bend=30,style={ultra thick, ultramarine, -{Latex[length=0.18cm,width=0.2cm]},  },loopsize=0.2cm,loopposition=90,loopshape=45,Direct](13)(2)

\Edge[,lw=0.2,bend=-8.531,style={semithick, -{Latex[length=0.09cm,width=0.1cm]},  },loopsize=0.2cm,loopposition=90,loopshape=45,Direct](2)(2)

\end{tikzpicture}
    \caption{Graph encoding of the circuit of \autoref{fig:exampleIQP}. The
    gadget labelled $A_k(\theta_k)$ results from the encoding of the
    $P_k(\theta_k)$ gate. The blue edges have weight $1$, and the adjacency
    matrices of the graph gadgets $A_k(\theta_k)$ were given earlier in
    \autoref{sec:example}. We encode the zeros in the input and output using a
    similar economical encoding as proposed in \cite[Note Ref.
    10]{rudolph_simple_2009}, \ie via self-loops of weight 1 (the white nodes).}
    \label{fig:graphExampleIQP}
\end{figure}

\section{Applications}\label{sec:applications}

\paragraph{Classical simulation.} The encoding constructs sparse graphs, even in the
general case, and even with gadget we do not know yet the design of, as it
constructs a matrix with blocks of constant size in its diagonal linked together
with a constant number of edges. As such, we can use the algorithm of
\cite{servedio_computing_2005}, which allows to compute the permanent of sparse
$n \times n$ matrices with $Cn$ nonzero entries in time $(2-\varepsilon)^n$ for
some strictly positive constant $\varepsilon$ and space $\bigo n$. This is
indeed better than Ryser's $\bigo{2^n}$ formula \cite{ryser_combinatorial_1973} for exact computation, or 
Gurvits' approximate algorithm
\cite{aaronson_generalizing_2012}, which, given
$A$ of order $n$, computes an estimate of $\per A$ to within
additive error $\pm \varepsilon \|A\|_2^n$ in time $O(n^2\varepsilon^{-2})$.
Because the cycles linking the gadgets and representing the variables must have
weight 1, we have $\|A\|_2 \geq 1$ \cite{hwang_cauchys_2004}. Incidentally, the connected components of the graph one obtains from the
encoding are correlated with the set of qubits interacting together. Recall that
for the matrix $A = \diag(A_2, \ldots, A_k)$ with $A_i$ a square matrix for all
$1 \leq i \leq k$,
\begin{equation}
\per{A} = \prod_{i = 1}^k \per{A_i}.
\end{equation}
This can be combined with \cite{servedio_computing_2005} to further speed up the
computation. However, in practice, for random $q$-qubit circuits of depth $d$
build upon 2-qubit gates, there are $\bigo{1}$ blocks as soon as ${d =
\bigo{\log q}}$.

\paragraph{Quantum simulation.}

A natural question that arises is how the method developed above compares against the standard gate-based implementation in linear optics following the KLM scheme \cite{knill_scheme_2001}. 
In \autoref{app:benchmark}, we benchmark the graph encoding against the KLM scheme and our encoding is found to be less demanding,
both in terms of samples and single photons per shot. 
We formally show that this holds and the results culminate in \autoref{th:sizeSq} in \autoref{app:shotsKLM} and \autoref{th:sizePq} in \autoref{sec:photons}.
The empirical bound of
$44$ qubits, where we expect improvement with almost unit probability, is
outside the scope of current simulators capabilities
\cite{heurtel_perceval_2023}. We believe similar bounds could be obtained for
other sets of gates, as the probabilistic nature of linear optical multi-qubit
gates plays in our favour.

\section{Conclusion}\label{sec:conclusion}

We have provided a method to compute sums of exponentiated polynomials as matrix permanents. This
technique finds its application in the encoding of quantum circuits amplitudes.
This yields a general extension of the encoding of quantum circuit output amplitudes
as graph permanents of \cite{rudolph_simple_2009}, unifying this graph-encoding
with polynomial expression of quantum circuit amplitudes as in
\cite{montanaro_quantum_2017}. To achieve this, we have shown that a clause
gadget can be found by solving a system of multilinear polynomial equations
through Gröbner basis theory. While hard in general, this computation has only
to be done once per degree, providing a parameterized gadget for the clauses of
that degree. Additionally, we remark that the
nice structure of the graph opens the door to permanent-based classical methods
for the simulation of quantum circuits
\cite{servedio_computing_2005,ryser_combinatorial_1973,jerrum_polynomial_2004}.

\section*{Acknowledgements}
The authors warmly thank \mbox{Ulysse Chabaud} and Shane Mansfield for the
valuable feedback on the manuscript and \mbox{Elham Kashefi} for the fruitful
discussions. This work has been co-funded by the European Commission as part of
the EIC accelerator program under the grant agreement 190188855 for SEPOQC
project, by the Horizon-CL4 program under the grant agreement 101135288 for
EPIQUE project.

\bibliographystyle{quantum}
\bibliography{bibliography}

\begin{thebibliography}{10}

\bibitem{maring_versatile_2024}
Nicolas Maring, Andreas Fyrillas, Mathias Pont, Edouard Ivanov, Petr Stepanov, Nico Margaria, William Hease, Anton Pishchagin, Aristide Lemaître, Isabelle Sagnes, Thi~Huong Au, Sébastien Boissier, Eric Bertasi, Aurélien Baert, Mario Valdivia, Marie Billard, Ozan Acar, Alexandre Brieussel, Rawad Mezher, Stephen~C. Wein, Alexia Salavrakos, Patrick Sinnott, Dario~A. Fioretto, Pierre-Emmanuel Emeriau, Nadia Belabas, Shane Mansfield, Pascale Senellart, Jean Senellart, and Niccolo Somaschi.
\newblock ``A versatile single-photon-based quantum computing platform''.
\newblock \href{https://dx.doi.org/10.1038/s41566-024-01403-4}{Nature Photonics}~(2024).

\bibitem{bartolucci_fusion-based_2023}
Sara Bartolucci, Patrick Birchall, Hector Bombín, Hugo Cable, Chris Dawson, Mercedes Gimeno-Segovia, Eric Johnston, Konrad Kieling, Naomi Nickerson, Mihir Pant, Fernando Pastawski, Terry Rudolph, and Chris Sparrow.
\newblock ``Fusion-based quantum computation''.
\newblock \href{https://dx.doi.org/10.1038/s41467-023-36493-1}{Nature Communications {\bf 14}, 912}~(2023).

\bibitem{aaronson_computational_2011}
Scott Aaronson and Alex Arkhipov.
\newblock ``The computational complexity of linear optics''.
\newblock In Proceedings of the forty-third annual {ACM} symposium on {Theory} of computing.
\newblock \href{https://dx.doi.org/10.1145/1993636.1993682}{Pages 333--342}.
\newblock San Jose California USA~(2011). ACM.

\bibitem{valiant_complexity_1979}
L.G. Valiant.
\newblock ``The complexity of computing the permanent''.
\newblock \href{https://dx.doi.org/10.1016/0304-3975(79)90044-6}{Theoretical Computer Science {\bf 8}, 189--201}~(1979).

\bibitem{bu_classical_2022}
Kaifeng Bu and Dax~Enshan Koh.
\newblock ``Classical {Simulation} of {Quantum} {Circuits} by {Half} {Gauss} {Sums}''.
\newblock \href{https://dx.doi.org/10.1007/s00220-022-04320-1}{Communications in Mathematical Physics {\bf 390}, 471--500}~(2022).

\bibitem{cai_tractable_2010}
Jin-Yi Cai, Xi~Chen, Richard Lipton, and Pinyan Lu.
\newblock ``On tractable exponential sums''.
\newblock In Frontiers in Algorithmics.
\newblock \href{https://dx.doi.org/10.1007/978-3-642-14553-7_16}{Pages 148--159}.
\newblock Springer~(2010).

\bibitem{bombieri_exponential_1966}
Enrico Bombieri.
\newblock ``On {Exponential} {Sums} in {Finite} {Fields}''.
\newblock \href{https://dx.doi.org/10.2307/2373048}{American Journal of Mathematics {\bf 88}, 71--105}~(1966).

\bibitem{korobov_weyls_1992}
N.~M. Korobov.
\newblock ``Weyl’s sums''.
\newblock \href{https://dx.doi.org/10.1007/978-94-015-8032-8_2}{Pages 68--138}.
\newblock Springer Netherlands. ~(1992).

\bibitem{stockmeyer_complexity_1983}
Larry Stockmeyer.
\newblock ``The complexity of approximate counting''.
\newblock In Proceedings of the fifteenth annual {ACM} symposium on Theory of computing - {STOC} '83.
\newblock \href{https://dx.doi.org/10.1145/800061.808740}{Pages 118--126}.
\newblock {ACM} Press~(1983).

\bibitem{feynman_quantum_1966}
R.~P. Feynman, A.~R. Hibbs, and George~H. Weiss.
\newblock ``\textit{Quantum Mechanics and Path Integrals}''.
\newblock \href{https://dx.doi.org/10.1063/1.3048320}{Physics Today {\bf 19}, 89--89}~(1966).

\bibitem{dawson_quantum_2004}
Christopher~M. Dawson, Andrew~P. Hines, Duncan Mortimer, Henry~L. Haselgrove, Michael~A. Nielsen, and Tobias~J. Osborne.
\newblock ``Quantum computing and polynomial equations over the finite field z2''.
\newblock \href{https://dx.doi.org/10.48550/arXiv.quant-ph/0408129}{Quantum Info. Comput. {\bf 5}, 102--112}~(2005).

\bibitem{ehrenfeucht_computational_1990}
Andrzej Ehrenfeucht and Marek Karpinski.
\newblock ``The computational complexity of (xor, and)-counting problems''.
\newblock \href{https://dx.doi.org/10.1109/9.400469}{International Computer Science Inst.} ~(1990).

\bibitem{rudolph_simple_2009}
Terry Rudolph.
\newblock ``Simple encoding of a quantum circuit amplitude as a matrix permanent''.
\newblock \href{https://dx.doi.org/10.1103/PhysRevA.80.054302}{Physical Review A {\bf 80}, 054302}~(2009).

\bibitem{gurvits_complexity_2005}
Leonid Gurvits.
\newblock ``On the complexity of mixed discriminants and related problems''.
\newblock In Joanna Jedrzejowicz and Andrzej Szepietowski, editors, Mathematical Foundations of Computer Science 2005.
\newblock \href{https://dx.doi.org/10.1007/11549345_39}{Pages 447--458}.
\newblock Springer Berlin Heidelberg~(2005).

\bibitem{aaronson_generalizing_2012}
Scott Aaronson and Travis Hance.
\newblock ``Generalizing and {Derandomizing} {Gurvits}'s {Approximation} {Algorithm} for the {Permanent}''~(2012).

\bibitem{blaser_complexity_2012}
Markus Bläser, Holger Dell, and Mahmoud Fouz.
\newblock ``Complexity and approximability of the cover polynomial''.
\newblock \href{https://dx.doi.org/10.1007/s00037-011-0018-0}{computational complexity {\bf 21}, 359--419}~(2012).

\bibitem{mezher_solving_2023}
Rawad Mezher, Ana~Filipa Carvalho, and Shane Mansfield.
\newblock ``Solving graph problems with single photons and linear optics''.
\newblock \href{https://dx.doi.org/10.1103/PhysRevA.108.032405}{Phys. Rev. A {\bf 108}, 032405}~(2023).

\bibitem{marcus_permanents_1965}
Marvin Marcus and Henryk Minc.
\newblock ``Permanents''.
\newblock \href{https://dx.doi.org/10.2307/2313846}{The American Mathematical Monthly {\bf 72}, 577}~(1965).

\bibitem{aaronson_linear-optical_2011}
Scott Aaronson.
\newblock ``A {Linear}-{Optical} {Proof} that the {Permanent} is \#{P}-{Hard}''.
\newblock \href{https://dx.doi.org/10.1098/rspa.2011.0232}{Proceedings of the Royal Society A: Mathematical, Physical and Engineering Sciences {\bf 467}, 3393--3405}~(2011).

\bibitem{cox_ideals_2013}
David Cox, John Little, and Donal {OShea}.
\newblock ``Ideals, varieties, and algorithms: an introduction to computational algebraic geometry and commutative algebra''.
\newblock \href{https://dx.doi.org/10.1007/978-3-319-16721-3}{Springer Science \& Business Media}. ~(2013).

\bibitem{faugere_new_2002}
Jean~Charles Faugère.
\newblock ``A new efficient algorithm for computing gröbner bases without reduction to zero (f5)''.
\newblock In Proceedings of the 2002 international symposium on Symbolic and algebraic computation.
\newblock \href{https://dx.doi.org/10.1145/780506.780516}{Pages 75--83}.
\newblock {ISSAC} '02. Association for Computing Machinery~(2002).

\bibitem{bardet_complexity_2015}
Magali Bardet, Jean-Charles Faugère, and Bruno Salvy.
\newblock ``On the complexity of the {F5} {Gröbner} basis algorithm''.
\newblock \href{https://dx.doi.org/https://doi.org/10.1016/j.jsc.2014.09.025}{Journal of Symbolic Computation {\bf 70}, 49--70}~(2015).

\bibitem{chen_homotopy_2015}
Tianran Chen and Tien-Yien Li.
\newblock ``Homotopy continuation method for solving systems of nonlinear and polynomial equations''.
\newblock \href{https://dx.doi.org/10.4310/CIS.2015.v15.n2.a1}{Communications in Information and Systems {\bf 15}, 119--307}~(2015).

\bibitem{berthomieu_msolve_2021}
Jérémy Berthomieu, Christian Eder, and Mohab Safey El~Din.
\newblock ``msolve: {A} {Library} for {Solving} {Polynomial} {Systems}''.
\newblock In 2021 {International} {Symposium} on {Symbolic} and {Algebraic} {Computation}.
\newblock \href{https://dx.doi.org/10.1145/3452143.3465545}{Pages 51--58}.
\newblock 46th {International} {Symposium} on {Symbolic} and {Algebraic} {Computation}Saint Petersburg, Russia~(2021). ACM.

\bibitem{demin_groebner_2023}
Alexander Demin and Shashi Gowda.
\newblock ``Groebner. jl: A package for gröbner bases computations in julia''.
\newblock \href{https://dx.doi.org/10.48550/arXiv.2304.06935}{{arXiv} preprint {arXiv}:2304.06935}~(2023).

\bibitem{bezanson_julia_2017}
Jeff Bezanson, Alan Edelman, Stefan Karpinski, and Viral~B. Shah.
\newblock ``Julia: {A} {Fresh} {Approach} to {Numerical} {Computing}''.
\newblock \href{https://dx.doi.org/10.1137/141000671}{SIAM Review {\bf 59}, 65--98}~(2017).

\bibitem{shepherd_temporally_2009}
Dan Shepherd and Michael~J. Bremner.
\newblock ``Temporally unstructured quantum computation''.
\newblock \href{https://dx.doi.org/10.1098/rspa.2008.0443}{Proceedings of the Royal Society A: Mathematical, Physical and Engineering Sciences {\bf 465}, 1413--1439}~(2009).

\bibitem{bremner_classical_2011}
Michael~J. Bremner, Richard Jozsa, and Dan~J. Shepherd.
\newblock ``Classical simulation of commuting quantum computations implies collapse of the polynomial hierarchy''.
\newblock \href{https://dx.doi.org/10.1098/rspa.2010.0301}{Proceedings of the Royal Society A: Mathematical, Physical and Engineering Sciences {\bf 467}, 459--472}~(2011).

\bibitem{bremner_average-case_2016}
Michael~J. Bremner, Ashley Montanaro, and Dan~J. Shepherd.
\newblock ``Average-case complexity versus approximate simulation of commuting quantum computations''.
\newblock \href{https://dx.doi.org/10.1103/physrevlett.117.080501}{Physical Review Letters{\bf 117}}~(2016).

\bibitem{montanaro_quantum_2017}
Ashley Montanaro.
\newblock ``Quantum circuits and low-degree polynomials over $\mathbb{F}_2$''.
\newblock \href{https://dx.doi.org/10.1088/1751-8121/aa565f}{Journal of Physics A: Mathematical and Theoretical {\bf 50}, 084002}~(2017).

\bibitem{servedio_computing_2005}
Rocco~A. Servedio and Andrew Wan.
\newblock ``Computing sparse permanents faster''.
\newblock \href{https://dx.doi.org/10.1016/j.ipl.2005.06.007}{Information Processing Letters {\bf 96}, 89--92}~(2005).

\bibitem{ryser_combinatorial_1973}
H.J. Ryser.
\newblock ``Combinatorial mathematics''.
\newblock \href{https://dx.doi.org/https://doi.org/10.5948/UPO9781614440147}{Carus mathematical monographs}. Mathematical Association of America. ~(1973).

\bibitem{hwang_cauchys_2004}
Suk-Geun Hwang.
\newblock ``Cauchy's {Interlace} {Theorem} for {Eigenvalues} of {Hermitian} {Matrices}''.
\newblock \href{https://dx.doi.org/10.2307/4145217}{The American Mathematical Monthly {\bf 111}, 157}~(2004).

\bibitem{knill_scheme_2001}
E~Knill, R~Laflamme, and G~Milburn.
\newblock ``A scheme for efficient quantum computation with linear optics''.
\newblock \href{https://dx.doi.org/10.1038/35051009}{Nature {\bf 409}, 46--52}~(2001).

\bibitem{heurtel_perceval_2023}
Nicolas Heurtel, Andreas Fyrillas, Grégoire~De Gliniasty, Raphaël Le~Bihan, Sébastien Malherbe, Marceau Pailhas, Eric Bertasi, Boris Bourdoncle, Pierre-Emmanuel Emeriau, Rawad Mezher, Luka Music, Nadia Belabas, Benoît Valiron, Pascale Senellart, Shane Mansfield, and Jean Senellart.
\newblock ``Perceval: {A} {Software} {Platform} for {Discrete} {Variable} {Photonic} {Quantum} {Computing}''.
\newblock \href{https://dx.doi.org/10.22331/q-2023-02-21-931}{Quantum {\bf 7}, 931}~(2023).

\bibitem{jerrum_polynomial_2004}
Mark Jerrum, Alistair Sinclair, and Eric Vigoda.
\newblock ``A polynomial-time approximation algorithm for the permanent of a matrix with nonnegative entries''.
\newblock \href{https://dx.doi.org/10.1145/1008731.1008738}{J. ACM {\bf 51}, 671–697}~(2004).

\bibitem{horn_matrix_2012}
Roger~A. Horn and Charles~R. Johnson.
\newblock ``Matrix analysis''.
\newblock \href{https://dx.doi.org/10.1017/CBO9780511810817}{Cambridge University Press}. Cambridge~(2012).
\newblock 2nd ed edition.

\bibitem{zheng_spectral_2008}
Baodong Zheng and Liancheng Wang.
\newblock ``Spectral radius and infinity norm of matrices''.
\newblock \href{https://dx.doi.org/10.1016/j.jmaa.2008.05.003}{Journal of Mathematical Analysis and Applications {\bf 346}, 243--250}~(2008).

\bibitem{carlson_minimax_1983}
David Carlson.
\newblock ``Minimax and interlacing thoerems for matrices''.
\newblock \href{https://dx.doi.org/10.1016/0024-3795(83)90211-2}{Linear Algebra and its Applications {\bf 54}, 153--172}~(1983).

\bibitem{knill_quantum_2002}
E.~Knill.
\newblock ``Quantum gates using linear optics and postselection''.
\newblock \href{https://dx.doi.org/10.1103/PhysRevA.66.052306}{Physical Review A {\bf 66}, 052306}~(2002).

\bibitem{hoeffding_probability_1963}
Wassily Hoeffding.
\newblock ``Probability {Inequalities} for {Sums} of {Bounded} {Random} {Variables}''.
\newblock \href{https://dx.doi.org/10.1080/01621459.1963.10500830}{Journal of the American Statistical Association {\bf 58}, 13--30}~(1963).

\bibitem{hofmann_quantum_2002}
Holger~F. Hofmann and Shigeki Takeuchi.
\newblock ``Quantum phase gate for photonic qubits using only beam splitters and postselection''.
\newblock \href{https://dx.doi.org/10.1103/PhysRevA.66.024308}{Physical Review A {\bf 66}, 024308}~(2002).

\bibitem{nielsen_quantum_2010}
Michael~A. Nielsen and Isaac~L. Chuang.
\newblock ``Quantum {Computation} and {Quantum} {Information}: 10th {Anniversary} {Edition}''.
\newblock \href{https://dx.doi.org/10.1017/CBO9780511976667}{Cambridge University Press}. ~(2010).

\bibitem{shende_cnot-cost_2008}
Vivek~V. Shende and Igor~L. Markov.
\newblock ``On the {CNOT}-cost of {TOFFOLI} gates''~(2008).
\newblock arXiv:0803.2316 [quant-ph].

\bibitem{ralph_efficient_2007}
T.~C. Ralph, K.~J. Resch, and A.~Gilchrist.
\newblock ``Efficient {Toffoli} {Gates} {Using} {Qudits}''.
\newblock \href{https://dx.doi.org/10.1103/PhysRevA.75.022313}{Physical Review A {\bf 75}, 022313}~(2007).

\bibitem{lin_quantum_2009}
Qing Lin and Jian Li.
\newblock ``Quantum control gates with weak cross-{Kerr} nonlinearity''.
\newblock \href{https://dx.doi.org/10.1103/PhysRevA.79.022301}{Physical Review A {\bf 79}, 022301}~(2009).

\bibitem{gottlieb_vc_2012}
Lee-Ad Gottlieb, Aryeh Kontorovich, and Elchanan Mossel.
\newblock ``{VC} bounds on the cardinality of nearly orthogonal function classes''.
\newblock \href{https://dx.doi.org/10.1016/j.disc.2012.01.030}{Discrete Mathematics {\bf 312}, 1766--1775}~(2012).

\end{thebibliography}

\onecolumn
\appendix

\section{Linear algebra}\label{app:linalg} 

This paragraph outlines the linear-algebra tools and notation used throughout
this work. We use bold fonts for a shorthand notation of vectors: $\x = (x_1,
\cdots, x_n)$. The singular value decomposition of a matrix $A \in \mathbb C^{m
\times n}$ is a factorization of $A$ into the product of three matrices $U
\Sigma V^*$, where the columns of $U$ and $V$ are orthonormal and $\Sigma =
\diag(\sigma_1, \cdots, \sigma_n)$ is a diagonal matrix with positive real
entries -- $A$'s \emph{singular values}. For a matrix $A \in \mathbb C^{m \times
n}$, $\|A\|_p$ denotes the matrix norm induced by the vector $p$-norm, for
$p\in[1, \infty]$, defined as
\begin{equation}
    \|A\|_p = \sup_{x \neq 0} \frac{\|Ax\|_p}{\|x\|_p}.
\end{equation}
Special cases of $p=1,2,\infty$ are 
\begin{subequations}
    \begin{align}
        \|A\|_{1\ } & =  \max_{1\leq j \leq n} \sum_{i=1}^m \abs{a_{ij}}, \\
        \|A\|_{2\ } & =  \sigma_{\max}(A), \label{eq:normSingularValue} \\ 
        \|A\|_\infty & =  \max_{1\leq i \leq m} \sum_{j=1}^n \abs{a_{ij}},
    \end{align}
\end{subequations}
where $\sigma_{\max}(A)$ denotes the largest singular value of $A$.  The
function $\|\cdot\|_p$ is sub-additive \cite{horn_matrix_2012}, \ie ${\|A+B\|_p
\leq \|A\|_p+\|B\|_p}$ for all $p\in[1, \infty]$. Let $\lambda_1, \ldots,
\lambda_n$ be the eigenvalues of $A$, the \emph{spectral radius} of $A$ is
defined as
\begin{equation}
    \rho(A) = \max \{\abs{\lambda_1}, \ldots,\abs{\lambda_n}\}.
\end{equation}
It follows that $\|A\|_2 = \rho(A^*A)$ where $A^*$ denotes the conjugate
transpose of $A$ and $\rho(A) \leq \|A\|_\infty$ for any real-valued matrix $A$
\cite{zheng_spectral_2008}. The notation $\diag$ is extended to matrices, \ie
the block diagonal matrix $A$ with blocks $A^{(1)}, \ldots, A^{(k)}$ is written
\begin{equation}
    A = \diag \left(\{A^{(j)}\}_{j=1}^k\right),
\end{equation}
with $A^{(j)}$ an $n_j \times n_j$ matrix and $\sum_{j=1}^k n_j=n$.

Finally, we will use
Cauchy's interlacing theorem for the eigenvalues of Hermitian matrices
\cite{carlson_minimax_1983}, which we states that, assuming $A\in \mathbb C^{n
\times n}$ is Hermitian and $B \in \mathbb C^{m \times m}$ with $m<n$ is a
principal submatrix of $A$, suppose $A$ has eigenvalues $\alpha_1 \geq \cdots
\geq \alpha_n$ and $B$ has eigenvalues $\beta_1\geq\cdots\geq\beta_m$, then, for
all $k = 1, \ldots, m,$ the eigenvalues of $B$ interlace those of $A$, \ie
\begin{equation}
    \alpha_k \geq \beta_k \geq \alpha_{k+n-m}.
\end{equation}

\section{Benchmark of the scheme and comparison with non-adaptive KLM}\label{app:benchmark}

Linear optical quantum computation involves $n$ identical photons evolving in
$m$ modes. Each computational basis state is of the form $\ket{\s} = (s_1,
\cdots, s_m)$ with $\sum_{i=1}^{m} s_i = n$, $s_i \geq 0$ is the number of
photons in mode $i$. The evolution of state $\s$ undergoing the transformation
induced by a linear optical circuit is described by a unitary operator $\mathcal
U$ acting on the state $\s$, we have
\begin{equation}
    \braket{\bs t | \mathcal U |\bs s} = \frac{\per{\mathcal U_{\bs s, \bs t}}}{\sqrt{\bs s! \bs t!}}.
\end{equation}

On the one hand, it was shown in \cite{mezher_solving_2023} that it is possible
to encode the adjacency matrix of a graph $G(V,E)$ onto a unitary matrix
$\mathcal U$ such that 
\begin{equation}
    \braket{\bs 1_n | \mathcal U |\bs 1_n} \propto \per{\mathcal G},
\end{equation}
where $\ket{\bs 1_n} =\ket{1^{\otimes n}0^{\otimes{m-n}}}$. In particular, let
$\CC$ be a quantum circuit and $G_\CC(V,E)$ be the graph encoding the amplitude
$\braket{\b | \CC | \a}$ for $\a, \b \in \{0, 1\}^n$, then it is possible to
encode the adjacency matrix $\GG_\CC$ into a linear optical circuit
$\mathcal V$ such that 
\begin{equation}
    \braket{\bs 1_n | \mathcal V |\bs 1_n} \propto \braket{\bs b | \CC | \bs a }.
\end{equation}

On the other hand, it was shown by Knill, Laflamme and Milburn (KLM)
\cite{knill_scheme_2001} that it is possible to implement any unitary
transformation using linear optical elements, single-photon sources and
single-photon detectors. The simplest, nonadaptive KLM scheme uses single photons and
passive linear optical elements; implementing multi-qubit gates is achieved
using auxiliary helper photons. It was shown that multi-qubit gates can only
performed probabilistically in that model, and the best known CZ gate
implementation succeeds with probability $2/27$ \cite{knill_quantum_2002}. The
implementation of the gate-based circuit $\mathcal C$ can be done directly using this
scheme. 

From either implementation, we can estimate the probability of interest by
sampling $N$ times from the output distribution of the linear optical device and
outputting the ratio $\frac{N_{post}}{N}$ where $N_{post}$ is the number of
times the post-selection condition is satisfied. 

In this section, we benchmark our method against the nonadaptive KLM scheme
\cite{knill_scheme_2001} on the task of estimating the zero-zero probability of
uniformly random $q$-qubit IQP circuits with diagonal gates \{Z, CZ, CCZ\}. We
first compare the amount of samples required to get an estimation
$\epsilon$-close from the true zero-zero probability, then give figures on the
amount of single photons needed to perform that task. We denote by $\II_q$ the
set of all $q$-qubit IQP circuits whose diagonal part is build upon the set
$\{Z, CZ, CCZ\}$, so that $\abs{\II_q} = 2^{\binom{q}{3} + \binom{q}{2} + q}$ \cite{montanaro_quantum_2017}.

\subsection{Post-selection rate}\label{sec:postSelectionRate}

In a physical experiment, the number of samples $N$ is finite and this induces a
statistical error in the estimation of the zero-zero probability. To give bounds
on that error, we use standard statistical inequality inequality
\cite{hoeffding_probability_1963}, which states that, assuming $X_1, \cdots,
X_N$ are $N$ i.i.d. random variables of expectation value $\mu := \e{X_i}$ such
that $a\leq X_i \leq b$ for all $i=1, \cdots, N$, then
\begin{equation}\label{eq:hoeffding}
    \pr{\left|\frac{1}{N}\sum_{i=1}^{N}X_i - \mu \right| \geq \epsilon} \leq \delta =  2e^{\frac{-2N\epsilon^2}{(b-a)^2}}.
\end{equation}
for any $0 < \epsilon < 1$. For simplicity, assume $a=b=1$. If $N$ is fixed, the
inequality can be used to bound the accuracy $\epsilon$ of the estimate for a
confidence parameter $\delta$, where $\epsilon$ reads
\begin{equation}\label{eq:epsilonHoeffding}
    \epsilon \leq \sqrt{\frac{1}{2N}\log\frac{2}{\delta}}.
\end{equation}

In order to apply Hoeffding's inequality to the two schemes, they must be
expressed in terms of the sum of i.i.d. random variables whose expectation
values -- hereafter denoted $\mu_\GG$ and $\mu_{\textsc{klm}}$ for the graph
encoding and KLM technique technique -- are known. While sampling the photonic
device, one has access to two quantities: $N$, the total number of samples, and
$N_{post}$, the number of \emph{post-selected} samples -- for the appropriate
definition of a post-selected sample. Therefore, Hoeffding's inequality can be
straightforwardly applied using the observed mean value $\frac{N_{post}}{N}$.
Depending on the scheme implemented (the graph encoding or KLM) on the photonic
devices that is sampled, we denote $\Ng$ or $\Nk$ the total number of samples
and $\Ngp$ or $\Nkp$ the number of post-selected samples.

\subsubsection{Post-selection for the graph-encoding technique.}
Given a $q$-qubit IQP circuit $\CC \in \II_q$, an estimate
\smash{$\est{\abs{\per{\GG_\CC}}^2}$} of the squared absolute value of the
graph's permanent is obtained by sampling the output distribution induced by the
$2M \times 2M$ unitary in which $G_\CC$ is embedded
\cite{mezher_solving_2023}, so that
\begin{equation}\label{eq:permanentEstimate}
    \est{\abs{\per{\GG_\CC}}^2} = \|\GG_\CC\|^{2M}_2 \frac{\Ngp}{\Ng},
\end{equation}
where $\Ngp$ is the number of samples with 1 photon in each of the $M$ first
modes. One is henceforth interested in the scaling of the ratio
$\frac{\Ngp}{\Ng}$, and in particular at what pace. The estimate we obtain from
\autoref{eq:permanentEstimate} is ultimately written
\begin{equation}
\est{\abs{\per{\GG_\CC}}^2} = \abs{\per{\GG_\CC}}^2\ \pm \ \delta(\Ng),
\end{equation}
where $\delta(\Ng)$ is such that
\begin{equation}
\lim_{\Ng\rightarrow +\infty} \delta(\Ng) = 0.
\end{equation}

First, we shall prove the following \autoref{pr:normUpperBound}, which gives a
constant upper bound on the largest singular value of $\GG_\CC$.

\begin{lemma}\label{pr:normUpperBound} Let $q, d \in
\mathbb N^*$. Let $\CC \in \II_q$ be a $q$-qubit IQP circuit of depth $d$ and
let $\GG_\CC$ be its graph encoding. Then,
\begin{equation}\label{eq:2normUpperBound}
\hspace{-.2cm}
\max\{1, \|A_i(\pi)\|_2 - 1\} \leq \|\GG_\CC\|_2 \leq \| A_i(\pi)\|_2 + 1,
\end{equation}
for the largest $i=0, 1, 2$ such that $\CC$ contains the gate $P_i(\pi)$.
\end{lemma}
\begin{proof}
    First, we show that for a block diagonal matrix $ A =
    \diag\left(\{ A^{(i)}\}_{i=1}^k\right)$, its $2$-norm is the maximum of the
    $2$-norms of its blocks. Using the identity of
    \autoref{eq:normSingularValue}, the singular value decomposition of $ A$ is
    performed individually on each block $ A^{(k)}$, yielding ${ A^{(k)} =
    U^{(k)}\Sigma^{(k)}V^{(k)*}}$. The singular values of $ A$ are therefore the
    union of the singular values of the blocks and in particular
    \begin{equation}\label{eq:normBlockMatrix}
        \|A\|_2
            = \max_{i} (\Sigma)_{ii}
            = \max_{k} \max_{j} (\Sigma^{(k)})_{jj}
            = \max_{k} \|A^{(k)}\|_2.
    \end{equation}

    For clarity, we write $A$ the adjacency matrix of $\GG_\CC$. The upper
    bound is obtained by writing ${A = A_g + A_c}$, where $A_g$ is the adjacency
    matrix of gadgets and $A_c$ the adjacency matrix of \emph{connectors}, \ie
    of edges corresponding the cycles encoding the variables as they connect
    together the gadgets. The matrix $A_g$ is a block diagonal matrix, where the
    blocks are the $A_i(\pi)$'s, and $A_c$ contains only $1$'s (and at most a
    single $1$ per row). Observe that
    $\|A_3(\pi)\|_2>\|A_2(\pi)\|_2>\|A_1(\pi)\|_2$. As such, the triangle
    inequality gives
    \begin{equation}
        \|A_g\|_2 - \|A_c\|_2 \leq  \|A\|_2 \leq \|A_g\|_2 + \|A_c\|_2.
    \end{equation}
    By definition, $A_g$ is a block diagonal matrix written $A_g =
    \diag\left(\{A_i^{(k)}(\pi)\}_{k=1}^n\right)$, with $n$ the number of gates
    in $\CC$. From \autoref{eq:normBlockMatrix}, $\|A_g\|_2$ = $\|A_i(\pi)\|_2$
    with the largest $i$ such that $\CC$ contains the gate $P_i(\pi)$. The
    matrix $A_c$ is, by construction and up to rows/columns of zeros, a
    permutation of the identity matrix. Thus $A_c$ is similar to a diagonal
    matrix with 0s and 1s (and it's not the null matrix) hence $\|A_c\|_2=1$.
    The $\max$ of the left-hand side of \autoref{eq:2normUpperBound} is obtained
    by observing that if $\CC$ contains no gate, then $A$, the adjacency matrix
    of $\GG_\CC$ is the idendity and therefore $\|A\|_2 = 1$.
\end{proof}
The numerical values of for $\|A_i(\pi)\|_2$ are -- rounded to two decimal
places -- 1, 1.73, and 3.53 for $i=0,1,2$ respectively. Hence, in all
generality, $1 \leq \|\GG_\CC\|_2 \leq 4.53 $ for all $\CC$. This result
utterly aborts the hope for efficient classical strong simulation raised at the
end of \cite{rudolph_simple_2009} since the norm of the adjacency matrix is
lower bounded by 1 which prevents the use of Gurvit's algorithm. In practice,
the numerical results suggest that $\|\GG_\CC\|_2$ grows rather slowly with
respect to $q$, see \autoref{fig:sigmaMaxAvg} and in particular it seems
\autoref{eq:2normUpperBound} is not tight. 

At first sight of \autoref{eq:permanentEstimate}, one could argue that the
singular values of the adjacency matrix can be lowered by considering the rescaled version
$\frac{1}{2^{\frac q M}}\GG_\CC$. This way,
\begin{equation}
    \frac{1}{2^q}\per{\GG_\CC} = \per{\frac{1}{2^{\frac q M}}\GG_\CC},
\end{equation}
therefore its norm becomes
\begin{equation}
\begin{aligned}
    \left\|\frac{1}{2^{\frac q M}}\GG_\CC\right\|_2
        = \frac{1}{2^{\frac{q}{M}}} \|\GG_\CC\|_2
        \leq \|\GG_\CC\|_2.
\end{aligned}
\end{equation}
However, $M = \bigo{qd}$. Assuming the circuits are not of constant depth, then
\begin{equation}
    \lim_{q \rightarrow \infty} \frac{1}{2^{\frac q M}} = 1,
\end{equation}
which means both graphs will have roughly the same norm. Moreover, the
finite statistics error will be the same in the two cases.

\begin{figure}[hbtp!]
    \centering
\begin{tikzpicture}[scale=.9]

    \definecolor{darkgray176}{RGB}{176,176,176}
    \definecolor{lightgray204}{RGB}{204,204,204}

    \begin{axis}[
    legend cell align={left},
    legend style={
      fill opacity=0.8,
      draw opacity=1,
      text opacity=1,
      at={(0.97,0.03)},
      anchor=south east,
      draw=lightgray204
    },
    tick align=outside,
    tick pos=left,
    x grid style={darkgray176},
    xlabel={\(\displaystyle q\)},
    xmin=-0.45, xmax=31.45,
    xtick style={color=black},
    y grid style={darkgray176},
    ymin=0, ymax=5.108088060406,
    ytick style={color=black}
    ]

    \addplot [draw=violet, fill=violet, mark=+, only marks, thick]
    table{%
    x  y
    1 1.10297349160595
    2 1.55182710281623
    3 2.71939794638946
    4 3.42456176047556
    5 3.78609022794239
    6 3.92811351565257
    7 3.98572214941026
    8 4.03005102343166
    9 4.05170053642777
    10 4.07047741714492
    11 4.08510677685372
    12 4.09346671181607
    13 4.10268217686957
    14 4.1103112931063
    15 4.11717389177935
    16 4.12296797069949
    17 4.12627467115837
    18 4.13289462740145
    19 4.13464170168475
    20 4.13869906678292
    21 4.14051466313216
    22 4.14535361108883
    23 4.14586713954592
    24 4.14725415173064
    25 4.15102586958022
    26 4.15249748788067
    27 4.15447448428148
    28 4.15499082522827
    29 4.15643304518945
    30 4.15753877775619
    };
    \addlegendentry{$\tilde{\|\GG(\CC_q)\|_2}$}

    \addplot [draw=ultramarine!80, fill=ultramarine!80, mark=-, only marks, thick]
    table{%
        x  y
        1 1.4142135623731
        2 2.27841360949644
        3 3.77574501839474
        4 4.09978075220876
        5 4.14158031244391
        6 4.1526747046905
        7 4.18362863492965
        8 4.1633249972181
        9 4.16082084093709
        10 4.18372368628686
        11 4.1915654191797
        12 4.17431758445109
        13 4.1766343337935
        14 4.20008847330494
        15 4.17658771363457
        16 4.20016230016505
        17 4.18513518237572
        18 4.18965989938496
        19 4.19367221582787
        20 4.18972564663077
        21 4.1993536433023
        22 4.18759553573544
        23 4.18116621898238
        24 4.19119730892492
        25 4.18907326797983
        26 4.18759069646778
        27 4.19267319916028
        28 4.19300239404689
        29 4.20525060737028
        30 4.19331703862073
    };
    \addlegendentry{$\max_{\CC_q}\|\GG(\CC_q)\|_2$}

    \addplot [semithick, darkred]
    table {%
    -0.45 4.53
    31.45 4.53
    };
    \addlegendentry{$\|A_3(\pi)\|_2 + 1$}

    \end{axis}

\end{tikzpicture}
    \caption{For each $1 \leq q \leq 30$, we computed the spectral norm of
    $5000$ graphs encoding the zero-zero amplitude of randomly chosen
    \emph{square} IQP-1 circuits (of depth ${d = q}$). For each $q$, $\tilde
    {\|\GG(\CC_q)\|_2}$ is the sample average of the spectral norm, the
    horizontal bar show the maximum value of $\|\GG(\CC_q)\|_2$ observed. All
    are indeed below the line $y = 4.53$ according to
    \autoref{pr:normUpperBound}.}
    \label{fig:sigmaMaxAvg}
\end{figure}
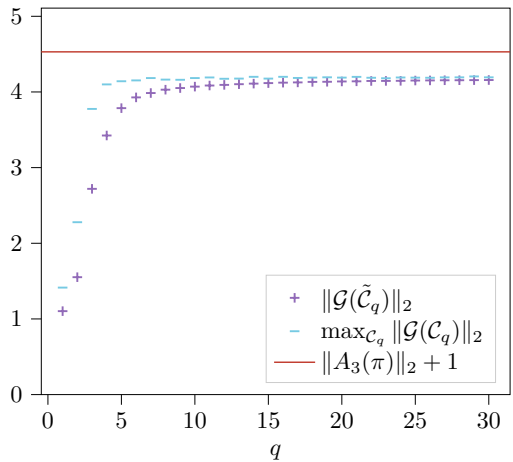

Using \autoref{eq:permanentEstimate} and according to the strong law of large
numbers, picking a uniformly random circuit $\CC \in \II_q$ and letting $\Ng$
going to infinity, define
\begin{equation}
    \mu_\GG := \lim_{\Ng\rightarrow \infty} \frac{\Ngp}{\Ng},
\end{equation}
with
\begin{equation}
    \lim_{\Ng\rightarrow \infty} \frac{\Ngp}{\Ng}\frac{\|\GG_\CC\|^{2M}_2}{2^{2q}}
    =     \abs{\braket{0^{\otimes q} | \CC | 0^{\otimes q}}}^2.
\end{equation}
It naturally follows from \autoref{eq:permanentEstimate} that
\begin{equation}
    \est{\abs{\braket{0^{\otimes q}|\CC|0^{\otimes q}}}^2}
    = \frac{\|\GG_\CC\|^{2M}_2}{2^{2q}}\frac{\Ngp}{\Ng}.
\end{equation}
We define the random variables $X_1, \cdots, X_N$, such that $X_i =
\frac{1}{2^{2q}}\|\GG_\CC\|^{2M}_2$ if we select the $i$-th sample, and $X_i =
0$ otherwise, so that
\begin{equation}
\frac{1}{\Ng}\sum_{i=1}^{\Ng} X_i = \est{\abs{\braket{0^{\otimes q}|\CC|0^{\otimes q}}}^2}.
\end{equation}
Accordingly, Hoeffding's inequality states that, defining
\begin{equation}\label{eq:epsGPrime}
\epsilon_\GG' := \frac{\|\GG_\CC\|^{2M}_2}{2^{2q}}\epsilon_\GG,
\end{equation}
the probability that the estimate of the zero-zero probability of a circuit
$\CC$ deviates from its actual value is
\begin{equation}\label{eq:prBadProbaEstimate}
    \pr{
            \left|\frac{1}{\Ng}\sum_{i=1}^{\Ng} X_i- \abs{\braket{0^{\otimes q}|\CC|0^{\otimes q}}}^2 \right|
            \geq \epsilon_\GG' 
        }
        \leq 2 e^{-2\Ng\epsilon_\GG^2}.
\end{equation}

\subsubsection{Post-selection for the KLM scheme.}
\label{app:shotsKLM}

In the same fashion, we can estimate the sampling accuracy of the implementation
of $\CC$ using the KLM scheme. Let ${p_{s,\CC} :=  p_{s,
\textsc{z}}^{\#(Z)}p_{s, \textsc{cz}}^{\#(CZ)}p_{s, \textsc{ccz}}^{\#(CCZ)}}$ be
the probability that all the gates of the circuit $\CC$ succeed. We have
\begin{equation}\label{eq:postSelectionRateKLMPs}
    \begin{aligned}
        \mu_{\textsc{klm}} := \lim_{\Nk\rightarrow \infty} \frac{\Nkp}{\Nk}
        =    \abs{\braket{0^{\otimes q} | \CC | 0^{\otimes q}} }^2 p_{s, \CC},
    \end{aligned}
\end{equation}
where, in this case, $\Nkp$ is the number of samples with zeros in all the $q$
logical qubits and with exactly one photon in each supplementary physical mode.
For on-chip linear optical computation, it is convenient to consider the
dual-rail encoding of the qubits. The major downside of linear optical quantum
computation is the intrinsicly probabilistic nature of multi-qubit operations,
which results in an exponential drop of the success probability of the
computation. Indeed $p_{s, \textsc{z}} = 1$ as a phase is applied using a single
phase shifter. We consider the Knill heralded CZ gate of
\cite{knill_quantum_2002}, which succeeds with probability $p_{s, \textsc{cz}}
:= \frac{2}{27}$. It requires two additional modes, each filled with a single
photon. The gate is applied successfully if and only if one photon is measured
at the output of each ancillary mode. A post-selected CZ gate, such as
\cite{hofmann_quantum_2002}, has a better success probability but necessitates
to post-select on having a precise pattern in the output modes which prevent the
composition of several post-selected gates. Also, note that already in
\cite{knill_scheme_2001} was proposed a technique to perform near-deterministic
CZ gates by utilizing large and highly entangled ancillary states and performing
state teleportation, but the creation of such resource states is far beyond the
reach of near-term hardware. As the induced resource overhead is nontrivial, we
consider access to merely ancillary single photons. The CCZ gate can be
performed using 6 CZ gates \cite{nielsen_quantum_2010}. The success probability
of the CCZ gate in linear optics is therefore $p_{s, \textsc{ccz}} :=
\left(\frac{2}{27}\right)^6$. It has been shown in previous work that $n$-qubit
Toffoli gates require at least $2n$ CNOT gates when implemented without
ancillary qubits \cite{shende_cnot-cost_2008}. Similarly, more
resource-efficient Toffoli gates have been proposed, at a cost of more involved
structures such as qudits \cite{ralph_efficient_2007} or cross-Kerr nonlinearity
\cite{lin_quantum_2009}, but we will restrict ourselves to the standard
Toffoli/CCZ implementation. Consequently, the success probability of the circuit
$\CC$, \ie the probability that all heralded modes are filled with exactly one
photon, is
\begin{equation}
    p_{s,\CC} = \left(\frac{2}{27}\right)^{\#(CZ)+6\#(CCZ)}.
\end{equation}
Define the random variables $Y_1, \cdots, Y_N$, such that $Y_i =1$ if the $i$-th
sample is post-selected, and $Y_i = 0$ otherwise. According to
\autoref{eq:postSelectionRateKLMPs}, finite experiments yield an estimate
\begin{equation}
    \frac{1}{\Nk}\sum_{i=1}^{\Nk} Y_i = \est{\abs{\braket{0^{\otimes q} | \CC | 0^{\otimes q}} }^2 p_{s, \CC}}.
\end{equation}
Since $p_{s, \CC}$ is known for a given circuit $\CC$ (or at least computable in
poly-time), the estimation of the zero-zero probability of a circuit $\CC$ one
obtains from sampling the photonic device can be written
\begin{equation}
    \est{\abs{\braket{0^{\otimes q}|\CC|0^{\otimes q}}}^2} =\frac{1}{p_{s,\CC}} \frac{\Nkp}{\Nk}.
\end{equation}
As in the previous subsection, we use the rescaled random variables ${Z_i =
\frac{1}{p_{s,\CC}}Y_i}$ to build the estimate. Similarly to
\autoref{eq:epsGPrime}, define
\begin{equation}
    \epsilon_{\textsc{klm}}' := \frac{\epsilon_{\textsc{klm}}}{{p_{s,\CC}}},
\end{equation}
and the error in the desired estimate is given by Hoeffding's inequality to be
\begin{equation}\label{eq:prBadProbaEstimateKLM}
    \hspace{-.3cm}\pr{
            \left|\frac{1}{\Nk}\sum_{i=1}^{\Nk}Z_i - \abs{\braket{0^{\otimes q}|\CC|0^{\otimes q}}}^2 \right|
            \geq  \epsilon_{\textsc{klm}}' 
        }
        \leq 2e^{-2\Nk\epsilon_{\textsc{klm}}^2}.
\end{equation}

To connect the two results, we shall prove the following
\autoref{th:schemesComparison}.
\begin{lemma}\label{th:schemesComparison} Let $q \in
    \mathbb N$. Let $\II_q$ be the set of $q$-qubit IQP circuits whose diagonal
    gates are drawn at random from the set $\{Z, CZ, CCZ\}$. Fix some desired
    accuracy $\epsilon_0$ and confidence $\delta$. Let $\Ngstar, \Nkstar \in
    \mathbb N$ be the number of samples needed to achieve an estimate
    $\epsilon_0$ close from the true probability with probability $1-\delta$
    using the graph encoding technique and the KLM scheme respectively. Then,
    there is a function $\alpha : \II_q \rightarrow \mathbb R$ of the number of
    gates and qubits of $\CC$, such that
    \begin{equation}
        \Ngstar = \alpha(\CC) \Nkstar.
    \end{equation}
\end{lemma}
\begin{proof}
    We applied Hoeffding's inequality to the two sampling strategies in
    \autoref{eq:prBadProbaEstimate,eq:prBadProbaEstimateKLM} to derive the
    number of samples required to estimate of the zero-zero probability of a
    given circuit to particular accuracy $\epsilon_\GG'$ and
    $\epsilon_{\textsc{klm}}'$. We write, thanks to
    \autoref{eq:epsilonHoeffding},
    \begin{equation}\label{eq:epsilonGG}
        \epsilon_\GG = \sqrt{\frac{1}{2\Ng}\log\frac{2}{\delta}},
    \end{equation}
    and
    \begin{equation}\label{eq:epsilonKLM}
        \epsilon_{\textsc{klm}} = \sqrt{\frac{1}{2\Nk}\log\frac{2}{\delta}}.
    \end{equation}
    Ultimately, one wants to find which of the two techniques requires the
    least number of samples to achieve a given accuracy $\epsilon_0$ with
    probability $1-\delta$. Precisely, this means finding values of
    $\Ngstar$ and $\Nkstar$ such that
    \begin{equation}
        \epsilon_0 = \epsilon_\GG' = \epsilon_{\textsc{klm}}'.
    \end{equation}
    Plugging \autoref{eq:epsilonGG,eq:epsilonKLM} into the previous equation
    yields
    \begin{equation}\label{eq:errorEqualityFromN}
        \frac{\|\GG_\CC\|^{2M}_2}{2^{2q}}\sqrt{\frac{1}{2\Ngstar}\log\frac{2}{\delta}}
            = \frac{1}{{p_{s,\CC}}}\sqrt{\frac{1}{2\Nkstar}\log\frac{2}{\delta}}.
    \end{equation}
    where we recall that
    \begin{equation} 
        M = q + \#(P_0(\theta)) + 3\#(P_1(\theta)) + 5\#(P_2(\theta))
    \end{equation}
    from \autoref{pr:numberOfNodes} and
    \begin{equation}
        p_{s,\CC} = \left(\frac{2}{27}\right)^{\#(CZ)+6\#(CCZ)}.
    \end{equation}
    Then \autoref{eq:errorEqualityFromN} simplifies to
    \begin{equation}
        \Ngstar= \frac{p_{s,\CC}^2\cdot\|\GG_\CC\|^{4M}_2}{2^{4q}} \Nkstar.
    \end{equation}
    Hence, define
    \begin{equation}
        \alpha(\CC) := \frac{p_{s,\CC}^2\cdot\|\GG_\CC\|^{4M}_2}{2^{4q}}
    \end{equation}
    and the proof is complete.
\end{proof}
The natural next step is to investigate the class of circuits $\CC$ for which
$\alpha(\CC) < 1$ implying that the graph encoding is more efficient to estimate
a given zero-zero probability than the KLM scheme with respect to the number of
samples requires. Pursuing this direction we have the following
\autoref{lm:cczCondition}.

\begin{lemma}\label{lm:cczCondition} Let $\CC \in
    \II_q$ be a $q$-qubit IQP circuit, and let $\alpha : \II_q \rightarrow
    \mathbb R$ as per \autoref{th:schemesComparison}. Recall that $\#(\UU)$ is
    defined as the number of occurrences the gates $\UU$ in $\CC$. Then,
    $\alpha(\CC) < 1$ if and only if
    \begin{equation}\label{eq:cczCondition}
        \#(CCZ) > \lceil 5.68\#(Z) + 12.14 \#(CZ) + 3.07 q \rceil.
    \end{equation}
\end{lemma}
\begin{proof}
    We write, for the sake of conciseness, $x,\ y$ and $z$ instead of $\#(Z),\
    \#(CZ)$ and $\#(CCZ)$ respectively. Recall from
    \autoref{th:schemesComparison} that for a given circuit $\CC$, writing
    ${\varsigma := \|\GG_\CC\|_2}$ and ${p:= \frac{2}{27}}$
    \begin{equation}
        \begin{aligned}
            \alpha(\CC)
                = \frac{p^{2y + 12z}\varsigma^{4q + 4x + 12y + 20z}}{2^{4q}}
                = \frac{
                        \varsigma^{4q}
                        \varsigma^{4x}
                        \left(p^{2}\varsigma^{12}\right)^y
                        \left(p^{12}\varsigma^{20}\right)^z
                    }{2^{4q}}.
        \end{aligned}
    \end{equation}
    From numerical calculation of the bound of \autoref{pr:normUpperBound},
    ${\varsigma \lesssim 4.53}$. Hence, in order to give figures, we assume that
    ${\varsigma = 4.53}$. We use the notation $\log_a(b)$ to denote the
    logarithm of $b$ in base $a$ and $\ln$ denotes the natural logarithm. Recall
    that we change the logarithm base from $a$ to $c$ using the identity
    \begin{equation}
        \log_a(b) = \frac{\log_c(b)}{\log_c(a)}.
    \end{equation}
    The condition $\alpha(\CC) < 1$ reads
    \begin{equation}\label{eq:alphaConditionRewritten}
        \frac{
            \varsigma^{4q}\varsigma^{4x}
            \left(p^{2}\varsigma^{12}\right)^y
            \left(p^{12}\varsigma^{20}\right)^z
        }{2^{4q}} < 1.
    \end{equation}
    Set $\gamma := p^{12}\varsigma^{20} < 1$,
    \autoref{eq:alphaConditionRewritten} rewrites
    \begin{equation}\label{eq:zLowerBoundLogGamma}
        z > - \log_\gamma\left(\left(\tfrac{\varsigma}{2}\right)^{4q}\right)
            - \log_\gamma\left(\varsigma^{4x}\right)
            - \log_\gamma\left(\left(p^{2}\varsigma^{12}\right)^{y}\right).
    \end{equation}
    We evaluate each term of the right-hand side of
    \autoref{eq:zLowerBoundLogGamma} separately. First,
    \begin{equation}\label{eq:approxQ}
        \log_\gamma\left(\left(\tfrac{\varsigma}{2}\right)^{4q}\right)
        = 4q\frac{\ln\varsigma - \ln 2}{\ln\gamma}
        \approx -3.07q.
    \end{equation}
    Second,
    \begin{equation}\label{eq:approxX}
        \log_\gamma\left(\varsigma^{4x}\right)
        = 4x\frac{\ln\varsigma}{\ln\gamma}
        \approx -5.68x.
    \end{equation}
    Last,
    \begin{equation}\label{eq:approxY}
        \log_\gamma\left(\left(p^{2}\varsigma^{12}\right)^{y}\right)
        = y\frac{2\ln p + 12\ln\varsigma}{\ln\gamma}
        \approx -12.14y.
    \end{equation}
    Plugging \autoref{eq:approxQ,eq:approxX,eq:approxY} into
    \autoref{eq:zLowerBoundLogGamma} concludes the proof.
\end{proof}
It remains to compute the probability that, for a fixed number of qubits $q$, a
uniformly random circuit $\CC \in \II_q$ satisfies the condition of the previous
\autoref{lm:cczCondition}, which summarizes \autoref{th:sizeSq}.

\begin{theorem}\label{th:sizeSq} Let $\CC$ be a $q$-qubit IQP,
    fix $\epsilon > 0$ and let $\Ng$, $\Nk$ the number of samples required to
    estimate $\braket{0^{\otimes q} | \CC | 0^{\otimes q}}$ using the graph
    encoding and the nonadaptive KLM scheme respectively. Then we can write 
    $
        {\Ng = \alpha(\CC) \Nk,}
    $
    and, for large enough $q$, we have
    $
        \pr{\alpha(\CC) < 1} \approx 1.
    $
\end{theorem}

\begin{proof}
    Let $S_q \subset \II_q$ be the set of circuits $\CC$ such that $\alpha(\CC)
    < 1$ and ${t(x, y, q) = \lceil 5.68x+ 12.14y + 3.07q \rceil}$ as per
    \autoref{eq:cczCondition}. We have,
    \begin{equation}
        \abs{S_q} =
            \sum_{x=0}^q
            \sum_{y=0}^{\binom{q}{2}}
            \sum_{z=t(x, y, q)}^{\binom{q}{3}}
            \binom{q}{x}\binom{\binom{q}{2}}{y}\binom{\binom{q}{3}}{z}.
    \end{equation}
    Since no closed form for the partial sum of the binomial coefficients is known,
    we use the following upper bound \cite{gottlieb_vc_2012}. Let $d \leq
    \frac{n}{2}$, then
    \begin{equation}\label{eq:binomialPartialSumUpperBound}
        \sum_{k=0}^d \binom{n}{k} \leq 2^{nH(d/n)},
    \end{equation}
    where $H$ is the binary entropy function defined as
    \begin{equation}
        H(x) := -x\log_2(x) - (1-x)\log_2(1-x).
    \end{equation}
    Consider $S^c_q$ the complement of $S_q$ in $\II_q$. Then,
    \begin{equation}\label{eq:complementS}
        \begin{aligned}
            \abs{S_q^c}
                & = \sum_{x=0}^q
                    \sum_{y=0}^{\binom{q}{2}}
                    \sum_{z=0}^{t(x, y, q)-1}
                    \binom{q}{x}\binom{\binom{q}{2}}{y}\binom{\binom{q}{3}}{z} 
                \leq \sum_{x=0}^q
                    \sum_{y=0}^{\binom{q}{2}}
                    \sum_{z=0}^{t\left(q, \binom{q}{2}, q\right) -1}
                    \binom{q}{x}\binom{\binom{q}{2}}{y}\binom{\binom{q}{3}}{z}. \\
        \end{aligned}
    \end{equation}
    In order to apply the upper bound of \autoref{eq:binomialPartialSumUpperBound},
    we need $q$ such that $t\left(q, \binom{q}{2}, q\right) -1 \leq \frac{1}{2}
    \binom{q}{3}$, \ie ${q \geq 77}$. Hence, fix ${q \geq 77}$,
    \autoref{eq:complementS} rewrites
    \begin{equation}
        \begin{aligned}
            \abs{S_q^c}
                \leq 2^{q + \binom{q}{2} + \binom{q}{3} H\left(\frac{t\left(q, \binom{q}{2}, q\right) -1}{\binom{q}{3}}\right)}.
        \end{aligned}
    \end{equation}
    Now, $t\left(q, \binom{q}{2}, q\right)$ is a quadratic function of $q$ while
    $\binom{q}{3}$ is a cubic function of $q$. Hence,
    \begin{equation}
        \lim_{q \rightarrow \infty} \frac{t\left(q, \binom{q}{2}, q\right) -1}{\binom{q}{3}} = 0,
    \end{equation}
    and therefore, the binary entropy function vanishes in the limit
    \begin{equation}
        \lim_{q \rightarrow \infty} H\left(\frac{t\left(q, \binom{q}{2}, q\right) -1}{\binom{q}{3}}\right) = 0.
    \end{equation}
    Finally, using the above, we conclude the proof
    \begin{equation}
        \lim_{q \rightarrow \infty} \frac{\abs{S_q^c}}{\abs{\II_q}} = 0.
    \end{equation}
\end{proof}

Said otherwise, our scheme requires less samples than nonadaptive KLM for
estimating a certain probability. This is supported by
\autoref{fig:prAlphaLessThanOne}, which shows numerical computation of
$\pr{\alpha(\CC) < 1}$ for uniformly random $q$-qubit IQP circuits. The
computations shown in \autoref{fig:prAlphaLessThanOne} were performed using the
Julia programming language \cite{bezanson_julia_2017}, which handles arbitrary
precision arithmetic. The numerical simulation indicates that already for $q =
44$,
\begin{equation}
    \pr{\alpha(\CC_{44}) < 1} \approx 1.
\end{equation}
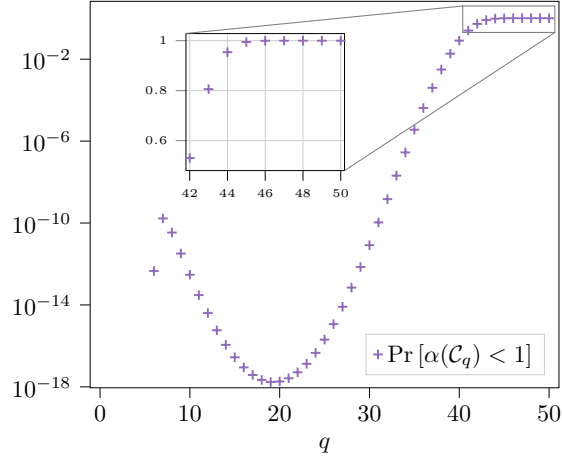
\begin{figure}[t]
\centering
\begin{tikzpicture}[scale=.9]

    \definecolor{darkgray176}{RGB}{176,176,176}
    \definecolor{lightgray204}{RGB}{204,204,204}

    \begin{axis}[
    legend cell align={left},
    legend style={
      fill opacity=0.8,
      draw opacity=1,
      text opacity=1,
      at={(0.03,0.97)},
      anchor=north west,
      draw=lightgray204
    },
    log basis y={10},
    tick align=outside,
    tick pos=left,
    x grid style={darkgray176},
    xlabel={\(\displaystyle q\)},
    xmin=-1.1, xmax=51.1,
    xtick style={color=black},
    y grid style={darkgray176},
    ymin=9e-19, ymax=6.69449721858594,
    ymode=log,
    ytick style={color=black},
    legend style={at={(0.97,0.14)},anchor=north east}
    ]
    \addplot [draw=violet, fill=violet, mark=+, only marks, thick]
    table{
    6 4.547e-13
    7 1.682e-10
    8 3.431e-11
    9 3.232e-12
    10 2.979e-13
    11 3.014e-14
    12 4.004e-15
    13 5.785e-16
    14 1.125e-16
    15 2.771e-17
    16 8.889e-18
    17 3.803e-18
    18 2.181e-18
    19 1.719e-18
    20 1.831e-18
    21 2.634e-18
    22 5.139e-18
    23 1.333e-17
    24 4.569e-17
    25 2.035e-16
    26 1.154e-15
    27 8.158e-15
    28 7.021e-14
    29 7.127e-13
    30 8.286e-12
    31 1.066e-10
    32 1.461e-09
    33 2.048e-08
    34 2.811e-07
    35 3.607e-06
    36 4.1265e-05
    37 0.000399
    38 0.003118
    39 0.018602
    40 0.081065
    41 0.248786
    42 0.530065
    43 0.806048
    44 0.954082
    45 0.994500
    46 0.999702
    47 0.9999935
    48 0.999999948
    49 0.99999999986
    50 0.9999999999999093
    };
    \addlegendentry{$\pr{\alpha(\CC_q) < 1}$}
    \path [draw=darkgray176]
       (axis cs:521.246753246753,401.215714291323)
    --(axis cs:564.190476190476,401.215714291323)
    --(axis cs:564.190476190476,406.539080989012)
    --(axis cs:521.246753246753,406.539080989012)
    --cycle;
    \end{axis}

    \begin{axis}[
    tick align=outside,
    tick pos=left,
    x grid style={lightgray204},
    xmin=41.8, xmax=50.2, 
    xmajorgrids, ymajorgrids,
    xtick style={color=black},
    y grid style={lightgray204},
    ymin=0.48, ymax=1.02846446110672,
    ticklabel style = {font=\tiny},
    width=3.9cm,
    height=3.6cm,
    at={(1.4cm, 3.2cm)}
    ]
    \addplot [draw=violet, fill=violet, mark=+, only marks, thick]
    table{%
    x y
    40 0.081065
    41 0.248786
    42 0.530065
    43 0.806048
    44 0.954082
    45 0.994500
    46 0.999702
    47 0.9999935
    48 0.999999948
    49 0.99999999986
    50 0.9999999999999093
    };
    \end{axis}
      \node[
          draw=gray,
          rectangle,
          minimum width=34.5pt,
          minimum height=10pt
      ] (a) at (6.12,5.43) {};
      \path [draw=gray] (1.4,5.22) -- (5.45,5.625);
      \path [draw=gray] (3.72,3.2) -- (6.798,5.237);
  \end{tikzpicture}
\caption{Log-linear plot of the numerical exact computation of the
probability that a uniformly random $q$-qubit IQP circuit satisfies
\autoref{eq:cczCondition}. The $y$-axis of the zoomed-in plot is in linear
scale.}
\label{fig:prAlphaLessThanOne}
\end{figure}

Recall that we denote by $M$ the size
of $\GG_\CC$ and equivalently the number of photons per shot required to
implement the probability estimation via \cite{mezher_solving_2023}. The
expected numbers of $Z$, $CZ$ and $CCZ$ gates in a uniformly random $q$-qubit
IQP circuit are $\frac{1}{2} q$, $\frac{1}{2}\binom{q}{2}$ and
$\frac{1}{2}\binom{q}{3}$ respectively. The expected number of photons per shot
required to estimate the zero-zero probability of a uniformly random $q$-qubit
IQP circuit using the graph encoding is therefore
\begin{equation}
   \e M = q + \frac{1}{2} \left(q + 3\binom{q}{2} + 5\binom{q}{3}\right) = \bigo{q^3}.
\end{equation}

\begin{table*}[t]
    \centering
    \captionsetup{width=.9\linewidth}
    \begin{tabular}{c|c||c|}
    \cline{2-3}
                                            & KLM                        & Graph  encoding                     \\ \hline
    \multicolumn{1}{|l|}{\# photons / shot}         & $q + 2\#(CZ) + 12\#(CCZ)$  & $q + \#(Z) + 3\#(CZ) + 5\#(CCZ)$    \\
    \multicolumn{1}{|l|}{\# modes}           & $2q + 2\#(CZ) + 12\#(CCZ)$ & $2q + 2\#(Z) + 6\#(CZ) + 10\#(CCZ)$ \\ \hline
    \end{tabular}
    \caption{Summary of the resources requirements of the KLM scheme and the
    graph encoding technique for estimating the zero-zero probability of a
    $q$-qubit IQP circuit with diagonal gates drawn at random from the set \{Z,
    CZ, CCZ\}.}
    \label{tbl:resourceRequirements}
\end{table*}

\subsection{Single-photons needs}\label{sec:photons}

The \autoref{tbl:resourceRequirements} summarizes the resources requirements of
the KLM scheme and the graph encoding technique for estimating the zero-zero
probability of a $q$-qubit IQP circuit with diagonal gates drawn at random from
the set $\{Z, CZ, CCZ\}$. An important parameter when dealing with linear
optics, and in particular near term devices, is the rate of $n$-photon events.
State-of-the-art on-chip photonic devices can produce $6$ photons
events at a rate of $\sim$4Hz
\cite{maring_versatile_2024}. This makes the number of
photons needed in the computation a critical parameter, and in particular a
factor utterly more limiting than the number of modes. To this extend, the graph
encoding technique is more profitable than the KLM scheme.

Let us define the two quantities ${K(\CC) := 2\#(CZ) + 12\#(CCZ)}$ and ${G(\CC)
:= \#(Z) + 3\#(CZ) + 5\#(CCZ)}$, where $\#(\UU)$ is the number of occurrences of
the gate $\UU$ in $\CC$. $K(\CC)$ and $G(\CC)$ quantify the number of photons
per shot (up to $q$) required to estimate the zero-zero probability of the
$q$-qubit IQP circuits associated $\CC$ using the KLM scheme and the graph
encoding technique respectively. If the photon coincidence rate is the limiting
factor of the computation, then the graph encoding is more suitable whenever
$G(\CC) < K(\CC)$, that is, whenever $\#(Z) + \#(CZ) < 7\#(CCZ)$. Conversely,
KLM is more profitable than the graph encoding technique as soon as $7\#(CCZ) <
\#(Z) + \#(CZ)$. The factor $7$ depends on the CCZ implementation we chose to
compare with, and might not be minimal. To that extend, we show that even if
were the case that the KLM implementation becomes more advantageous when
$\#(CCZ) < \#(Z) + \#(CZ)$, the fraction of such IQP circuits is negligible
compared to the total number of IQP circuits. 

The point of this short section is the following \autoref{{th:sizePq}}.

\begin{theorem}\label{th:sizePq} Let $P_q$ be the set of $q$-qubit IQP circuits
    for which the estimation of the zero-zero probability via the graph encoding
    technique requires fewer photons per shot. Then, 
    \begin{equation}
        \lim_{q \to \infty} \frac{\abs{P_q}}{\abs{\II_q}} = 1.
    \end{equation}
\end{theorem}
\begin{proof}
    The proof follows the lines of that of \autoref{th:sizeSq}. Let $P_q \subset
    \CC_q$ be the set of $q$-qubit IQP circuits such that ${\#(CCZ) > \#(Z) +
    \#(CZ)}$ and let $P_q^c \subset \CC_q$ be the complement of $P_q$ in $\II_q$,
    that is, the set of $q$-qubit IQP circuits for which the estimation of the
    zero-zero probability via the graph encoding technique requires more photons
    per shot. Precisely, $P_q^c$ is the set of $q$-qubit IQP circuits such that
    ${\#(CCZ) < \#(Z) + \#(CZ)}$, whose cardinality is given by
    \begin{equation}
    \begin{aligned}
    \abs{P_q^c} & = \sum_{x = 0}^q
        \sum_{y = 0}^{\binom{q}{2}}
        \sum_{z = 0}^{x+y - 1} \binom{q}{x} \binom{\binom{q}{2}}{y} \binom{\binom{q}{3}}{z} 
        \leq \sum_{x = 0}^q
        \sum_{y = 0}^{\binom{q}{2}}
        \sum_{z = 0}^{\binom{q}{2} + q} \binom{q}{x} \binom{\binom{q}{2}}{y} \binom{\binom{q}{3}}{z}. \\
    \end{aligned}
    \end{equation}
    To use the upper bound of \autoref{eq:binomialPartialSumUpperBound}, we need to
    find $q$ such that ${\binom{q}{2} + q \leq \frac 1 2 \binom{q}{3}}$, that is,
    ${q \geq 10}$. Hence, 
    \begin{equation}
        \abs{P_q^c} \leq 2^{ q + \binom{q}{2} + H\left(\frac{q + \binom{q}{2}}{\binom{q}{3}}\right) \binom{q}{3}},
    \end{equation}
    and as such, as $q$ grows, the probability that a uniformly random $q$-qubit IQP
    circuit is such that ${\#(CCZ) < \#(Z) + \#(CZ)}$ is given by
    \begin{equation}
    \begin{aligned}
        \lim_{q \to \infty} \frac{\abs{P_q^c}}{\abs{\II_q}}
    & \leq \lim_{q \to \infty}  \frac{2^{q + \binom{q}{2} + H\left(\frac{q + \binom{q}{2}}{\binom{q}{3}}\right) \binom{q}{3}}}{2^{q + \binom{q}{2} + \binom{q}{3}}}
    = 0.
    \end{aligned}
    \end{equation}
\end{proof}
In other words, ${\#(CCZ) > \#(Z) + \#(CZ)}$ almost surely for large enough $q$
and the graph encoding technique requires fewer photons per shot.

\end{document}